\pdfoutput=1
\documentclass[a4paper,reprint,onecolumn,notitlepage,aip,cha,nofootinbib]{revtex4-2}
\usepackage{geometry}

\usepackage[english]{babel}
\usepackage[T1]{fontenc}
\usepackage{tgtermes} 
\usepackage{tgheros} 
\renewcommand{\selectlanguage}[1]{} 

\usepackage{amsmath}
\usepackage{amssymb}
\usepackage{amsthm}
\usepackage{physics}
\usepackage{mathtools}

\usepackage{bm}
\usepackage{enumerate}
\usepackage{chngcntr}
\usepackage[dvipsnames]{xcolor}
\usepackage{bbold}
\usepackage[makeroom]{cancel}
\usepackage{soul}

\usepackage[normalem]{ulem}

\usepackage{mathrsfs}

\usepackage{hyperref}
\usepackage[capitalise]{cleveref}
\usepackage{tikz}
\usepackage{tikz-3dplot}
\usepackage{tikz-cd} 
\usetikzlibrary{patterns}

\usepackage{subcaption}

\DeclareMathOperator{\ran}{ran}
\DeclareMathOperator{\Span}{Span}

\crefname{equation}{}{}

\newtheorem{theorem}{Theorem}[section]
\newtheorem{lemma}[theorem]{Lemma}
\crefname{lemma}{Lemma}{Lemmas}
\Crefname{lemma}{Lemma}{Lemmas}

\crefname{corollary}{corollary}{Corollaries}
\Crefname{corollary}{Corollary}{Corollaries}
\newtheorem{proposition}[theorem]{Proposition}
\crefname{proposition}{Proposition}{Propositions}
\Crefname{proposition}{Proposition}{Propositions}

\theoremstyle{remark}

\crefname{remark}{Remark}{Remarks}
\Crefname{remark}{Remark}{Remarks}

\newtheorem{example}{Example}
\crefname{example}{Example}{Examples}
\Crefname{example}{Example}{Examples}

\theoremstyle{definition}

\crefname{definition}{Def.}{Defs.}
\Crefname{definition}{Definition}{Definitions}

\DeclareMathOperator{\diag}{diag}
\DeclareMathOperator{\Aff}{Aff}

\newcommand{\RR}{\mathbb{R}}
\newcommand{\CC}{\mathbb{C}}
\newcommand{\wconv}{\rightharpoonup}
\newcommand{\iden}{\mathbb{1}}
\renewcommand{\i}{\mathrm{i}}

\newcommand{\LL}{\mathrm{LL}}
\newcommand{\LI}{\mathrm{L}}

\newcommand{\dua}[2]{\langle #1, #2 \rangle}

\newcommand{\sps}[1]{\mathring{#1}}
\renewcommand{\ol}[1]{\overline{#1}}
\newcommand{\wt}[1]{\widetilde{#1}}

\newcommand{\MC}{\mathcal{M}}
\newcommand{\TC}{\mathcal{T}}
\newcommand{\HC}{\mathcal{H}}
\newcommand{\DC}{\mathcal{D}}

\newcommand{\EC}{\mathcal{E}}

\newcommand{\RC}{\mathcal{R}}

\renewcommand{\phi}{\varphi}

\newcommand{\valpha}{{\bm{\alpha}}}
\newcommand{\vbeta}{{\bm{\beta}}}
\newcommand{\vpsi}{{\bm{\psi}}}
\newcommand{\vphi}{{\bm{\phi}}}
\newcommand{\vsigma}{{\bm{\sigma}}}
\newcommand{\vzeta}{{\bm{\zeta}}}

\newcommand{\vchi}{{\bm{\chi}}}
\newcommand{\vxi}{{\bm{\xi}}}

\newcommand{\vgamma}{{\bm{\gamma}}}

\newcommand{\vomega}{{\bm{\omega}}}

\newcommand{\vmu}{{\bm{\mu}}}

\newcommand{\vGamma}{{\bm{\Gamma}}}
\newcommand{\vLambda}{{\bm{\Lambda}}}

\newcommand{\vOmega}{{\bm{\Omega}}}
\renewcommand{\va}{\mathbf{a}}
\newcommand{\ve}{\mathbf{e}}
\renewcommand{\vu}{\mathbf{u}}

\newcommand{\vw}{\mathbf{w}}
\newcommand{\vy}{\mathbf{y}}
\newcommand{\vc}{\mathbf{c}}

\newcommand{\vh}{\mathbf{h}}
\newcommand{\vx}{\mathbf{x}}
\newcommand{\lapl}{\bm{\Delta}}
\newcommand{\vt}{\mathbf{t}}
\newcommand{\vH}{\mathbf{H}}
\newcommand{\vv}{\mathbf{v}}
\newcommand{\vj}{\mathbf{j}}
\newcommand{\vf}{\mathbf{f}}
\newcommand{\vV}{\mathbf{V}}
\newcommand{\vT}{\mathbf{T}}
\newcommand{\vK}{\mathbf{K}}
\newcommand{\vM}{{\mathbf{M}}}
\newcommand{\vA}{{\mathbf{A}}}
\newcommand{\vW}{{\mathbf{W}}}

\newcommand{\vP}{{\mathbf{P}}}

\begin{document}

\title[DFT for the Dicke Hamiltonian]{Density-functional theory for the Dicke Hamiltonian}

\author{Vebj{\o}rn H. Bakkestuen}
\affiliation{Department of Computer Science, Oslo Metropolitan University, Norway}

\author{Mih\'aly A. Csirik}
\affiliation{Department of Computer Science, Oslo Metropolitan University, Norway}
\affiliation{Hylleraas Centre for Quantum Molecular Sciences, Department of Chemistry, University of Oslo, Norway}

\author{Andre Laestadius}
\affiliation{Department of Computer Science, Oslo Metropolitan University, Norway}
\affiliation{Hylleraas Centre for Quantum Molecular Sciences, Department of Chemistry, University of Oslo, Norway}
\email{andre.laestadius@oslomet.no}

\author{Markus Penz}
\affiliation{Max Planck Institute for the Structure and Dynamics of Matter, Hamburg, Germany}
\affiliation{Department of Computer Science, Oslo Metropolitan University, Norway}

\begin{abstract}
A detailed analysis of density-functional theory for quantum-elec\-tro\-dynamical model systems is provided.  
In particular, the quantum Rabi model, the Dicke model, and a generalization of the latter to multiple modes are considered.
We prove a Hohenberg--Kohn theorem that manifests the magnetization and displacement as internal variables, along with several representability results. The constrained-search functionals for pure states and ensembles are introduced and analyzed.
We find the optimizers for the pure-state constrained-search functional to be low-lying eigenstates of the Hamiltonian and, based on the properties of the optimizers, we formulate an adiabatic-connection formula.
In the reduced case of the Rabi model we can even show differentiability of the universal density functional, which amounts to unique pure-state $v$-representability.
\end{abstract}


\maketitle

\section{Introduction}

Quantum electrodynamics (QED) is the fully quantized theory of matter and light~\cite{ryder1996quantum,greiner2013field}. It describes the interaction between charged particles through their coupling to the electromagnetic field. Apart from high-energy physics, particularly in the domain of equilibrium condensed-matter physics, non-relativistic QED in the shape of the Pauli--Fierz Hamiltonian~\cite{spohn2004dynamics} is considered sufficient to describe interesting effects, such as the modification of chemical and material properties~\cite{ebbesen2016hybrid,garcia2021manipulating,Ruggenthaler2023}. In order to explain those, and due to the high number of involved particles and ensuing complexity of the problem, well-established first-principle approximation methods, such as density-functional theory (DFT), were adapted for non-relativistic QED~\cite{Ruggenthaler2014,Ruggenthaler2017,Ruggenthaler2018,flick2019light,jestadt2019light,konecny2024relativistic}. DFT itself is an approximation technique for many-particle quantum systems, ubiquitous in chemistry and materials science, where the correlated wavefunction is replaced with a reduced, collective variable~\cite{eschrig2003-book,dreizler2012-book,cances-2023-DFT-book}. In the standard formulation this variable is the particle density that, following a famous idea by Hohenberg and Kohn~\cite{Hohenberg1964}, maps to the unique external potential that then again allows to retrieve the wavefunction as the ground state of the Schrödinger equation. Different formulations of DFT allow different levels of mathematical rigor in this construction~\cite{penz2023-Review-I,penz2023-Review-II} and we aim at proving the mathematical cornerstones of DFT for one relatively simple extension of DFT to QED (QEDFT) here.

While most work in QEDFT is based on the Pauli--Fierz Hamiltonian, various approximations to this Hamiltonian are used as starting points for further investigations. These reduced Hamiltonians lead to a hierarchy of QEDFTs and yield a direct connection to well-established models of quantum optics that are designed to describe the photonic subsystem accurately while strongly simplifying the matter part. One such paradigmatic quantum-optical model is the quantum Rabi model. Despite its physical simplicity, it consists of a single two-level system coupled to one photonic mode, mathematically it is a highly non-trivial problem and only relatively recently an analytical expression for its spectrum has been found based on a Bargmann-space reformulation~\cite{Braak2011} (also see the reviews on the topic~\cite{Braak2015,Xie2017}). The same study sparked a controversy over the integrability of the system (and integrability in general)~\cite{Moroz2013,Batchelor2015-Integrability,Cunha2016} and the interest of the mathematics community in the model continues as of today~\cite{Nguyen2024,Hiroshima2025}.  Similar mathematical results have also been achieved for the Dicke model~\cite{Braak2013-Dicke,He2015-Dicke}, which describes multiple two-level systems coupled to a photonic mode and which recently received a QEDFT-type formulation~\cite{Novokreschenov2023}. In this work, that aims at a mathematical formulation of QEDFT for a simple light-matter system, we focus on a generalization of the Dicke model that also allows multiple photonic modes. Due to the relative simplicity of the matter subsystem in the Dicke Hamiltonian, this model allows to focus on the novel aspect of QEDFT, which correlates two physically different subsystems, those of matter and light. We note that this model includes the standard quantum-optical simplification that disregards the all-to-all dipole interaction that arises in the long-wavelength limit of the Pauli--Fierz Hamiltonian~\cite{Ruggenthaler2023}. This term raised a lot of interest in the last years~\cite{rokaj2018light,schaefer2020relevance,schnappinger2023cavity,borges2024extending}, since it can potentially explain the modifications of chemical and material equilibrium properties observed in collective-coupling situations~\cite{sidler2024unraveling}. This form of the light-matter Hamiltonian with an additional interaction will be investigated in a separate work.

One objective of the paper at hand is the extension of Lieb's analysis of standard Coulombic DFT~\cite{Lieb1983} to a model in QEDFT. 
It includes further techniques from more recent approaches based on convex analysis~\cite{Laestadius2024}. Due to the reduced complexity in the model, we are able to achieve considerably more than has been possible for the standard theory so far. This mainly includes results concerning ``$v$-representability'' and many interesting properties of the universal density functionals.

\cref{sec:prelim} first introduces the \emph{multi-mode Dicke model} and the relevant notation. \cref{sec:main} then assembles the main results, starting with a \emph{Hohenberg--Kohn theorem} in \cref{sec:HK}. For a ground state this theorem proves the unique mapping from the magnetization and displacement vectors (as the density variables) to the external potentials. Yet, this comes with an important restriction: a measure-zero set of magnetizations cannot be uniquely mapped, only those that are \emph{regular} by our definition. The other important feature of DFT is the \emph{Levy--Lieb (constrained-search) functional} defined and discussed in \cref{llsec}. This introduces the \emph{constraint manifold}, the set of all wavefunctions with a given magnetization and displacement, into the discussion. The Levy--Lieb functional is an optimization problem on this constraint manifold, and its optimizers are demonstrated to satisfy the Schrödinger equation, while not necessarily being ground states, yet low-lying eigenstates. \cref{sec:AC} shows how the \emph{adiabatic connection} can be constructed for the given model. We then proceed to the \emph{Lieb functional} in \cref{flsec} that extends the search space to mixed states. Finally, the model is reduced to a single two-level system coupled to one photon mode, the \emph{quantum Rabi model}, in \cref{specsec}. This simplification allows to achieve (i) a one-to-one mapping between the density variables and the potentials (except at the boundary values for magnetization), (ii) equality between the Levy--Lieb functional and the Lieb functional, and (iii) differentiability of this functional. These are all properties that are necessary for the formulation of DFT and such we are able to fully extend this theory to a restricted QED setting.

\subsection*{Acknowledgements} 
The authors would like to express their gratitude towards Michael Ruggenthaler for valuable guidance and feedback.
AL, MACs, MP, VHB were supported by ERC-2021-STG grant agreement No. 101041487 REGAL. 
AL and MACs were also supported by Research Council of Norway through funding of the CoE Hylleraas Centre for Quantum Molecular Sciences Grant No. 262695 and CCerror Grant No. 287906.

\section{Preliminary notions}
\label{sec:prelim}

\subsection{Notations and function spaces}

In this work, we consider a set of $N$ two-level systems (matter part) individually coupled to $M$ modes of a quantized radiation field (light part). The latter are conveniently described as quantum harmonic oscillators and the corresponding Hilbert space is thus $\HC=\HC_{\rm ph}\otimes\HC_{\rm f}$, where $\HC_{\rm ph}=\bigotimes^M L^2(\RR)$ and 
$\HC_{\rm f}=\bigotimes^N \CC^2 \simeq \CC^{2^N}$. We have
$$
\HC\simeq L^2(\RR^M)\otimes \CC^{2^N} \simeq L^2(\RR^M, \CC^{2^N}).
$$
Here, as usual, $L^2(\RR)$ is the Hilbert space of square-integrable complex-valued functions, equipped with the usual inner product $\dua{\cdot}{\cdot}$,
conjugate-linear in its first argument; the norm on this space is denoted by $\|\cdot\|$. We use the same notations on $L^2(\RR^M,\CC^{2^N})$,
i.e., 
$$
\dua{\vphi}{\vpsi}=\sum_{\valpha} \dua{\phi^\valpha}{\psi^\valpha} 
= \sum_{\alpha_1,\ldots,\alpha_N\in\{+,-\}} \int_{\RR^M} \ol{\phi^{\alpha_1,\ldots,\alpha_N}(\vx)} \psi^{\alpha_1,\ldots,\alpha_N}(\vx)\,\dd\vx,
$$
where $\psi^\valpha$ is the spin projection of $\vpsi$ corresponding to the eigenvector of the lifted Pauli matrices $\sigma_z^j$ indexed by 
the multiindex $\valpha\in\{+,-\}^N$.
Here, for any $j=1,\ldots,N$, we have set
$$
\sigma_a^j=\iden \otimes \ldots  \otimes \iden \otimes \underbrace{\sigma_a}_{\text{$j$th}} \otimes \iden \otimes \ldots \iden \in \CC^{2^N\times 2^N},
$$
where the Pauli matrices are given by
$$
\sigma_x=\begin{pmatrix}
0 & 1\\
1 & 0
\end{pmatrix}, \quad 
\sigma_y=\begin{pmatrix}
0 & -\i\\
\i & 0
\end{pmatrix}, \quad \text{and} \quad 
\sigma_z=\begin{pmatrix}
1 & 0\\
0 & -1
\end{pmatrix}.
$$
For convenience, we introduce the vector of lifted Pauli matrices,
\begin{equation*}
\vsigma_a=(\sigma_a^1, \ldots, \sigma_a^N)^\top \in \left(\CC^{2^N \times 2^N}\right)^N.
\end{equation*}
For instance, if $N=2$, then 
$$
\vsigma_z=\left(
\begin{pmatrix}
1 & & &\\
& 1 & &\\
& & -1 &\\
& & & -1
\end{pmatrix},
\begin{pmatrix}
1 & & &\\
& -1 & &\\
& & 1 &\\
& & & -1
\end{pmatrix}
\right)^\top,
$$
which has always diagonal form, and
$$
\vsigma_x=\left(
\begin{pmatrix}
0& 0 & 1 & 0\\
0& 0& 0 & 1\\
1& 0 & 0 & 0 \\
0& 1& 0 & 0
\end{pmatrix},
\begin{pmatrix}
0& 1 & 0 & 0\\
1& 0& 0 & 0\\
0& 0 & 0 & 1 \\
0& 0& 1 & 0
\end{pmatrix}
\right)^\top.
$$

\subsection{Multi-mode Dicke Hamiltonian}
\label{sec:Ham}

We first introduce the ``internal'' part of our Hamiltonian $\vH_0 : \HC\to\HC$, given by
\begin{equation}\label{h0def}
\vH_0 = (-\lapl_{\RR^M} + |\vx|^2 )\iden_{\CC^{2^N}} +  \vx\cdot\vLambda \vsigma_z - \vt\cdot \vsigma_x    
\end{equation}
which is essentially self-adjoint on $C_c^\infty(\RR^M, \CC^{2^N})$ as an operator on $\HC$.
Here,  $\lapl_{\RR^M}$ is the usual Laplace operator on $\RR^M$, which we will henceforth simply denote as $\lapl$.
Also, $\vLambda\in \RR^{M\times N}$ and $\vt\in\RR^N$ with $\vt\neq\bm{0}$. The product $\vLambda\vsigma_z$ is to be understood as the $M$-vector of $2^N\times 2^N$ matrices 
$$
\vLambda\vsigma_z=\left(\sum_{n=1}^N\Lambda_{1n}\sigma_z^n,\ldots,\sum_{n=1}^N\Lambda_{Mn}\sigma_z^n\right)^\top.
$$
In the form above, we recognize the Hamiltonian as a variant of the harmonic oscillator with coordinates $\vx$ and non-commuting coefficients, and indeed there is a connection to the field of ``non-commutative harmonic oscillators''~\cite{Wakayama2015}.
We will usually suppress the $\iden_{\CC^{2^N}}$ acting on the two-level systems.
Then, without vector notation, the Hamiltonian reads
$$
\vH_0=\sum_{m=1}^M (-\partial_{x_m}^2 + x_m^2) + \sum_{m=1}^M\sum_{n=1}^N x_m\Lambda_{mn}\sigma_z^n - \sum_{n=1}^N t_n\sigma_x^n.
$$
We may write
\begin{align*}
\left|\vx + \frac{1}{2} \vLambda \vsigma_z\right|^2 
&= |\vx|^2 + \vx \cdot \vLambda \vsigma_z +  \frac{1}{4} \vsigma_z \cdot (\vLambda^\top \vLambda \vsigma_z),
\end{align*}
so that
$$
\vH_0 = -\lapl + \vV - \vt\cdot \vsigma_x - \frac{1}{4} \vsigma_z \cdot (\vLambda^\top \vLambda \vsigma_z),
$$
where $\vV(\vx)=\left(\vx + \frac{1}{2} \vLambda \vsigma_z\right)^2$.
This shows that $\vH_0$ is bounded from below,
\begin{align*}
\dua{\vpsi}{\vH_0\vpsi}&=\int_{\RR^M} \left(|\grad\vpsi|^2 + \left|\vV\vpsi\right|^2\right) - \dua{\vpsi}{\vt\cdot \vsigma_x\vpsi} 
- \frac{1}{4} \dua{\vsigma_z\vpsi}{\vLambda^\top \vLambda \vsigma_z\vpsi}\\
&\ge \int_{\RR^M} \left(|\grad\vpsi|^2 + \left|\vV\vpsi\right|^2\right) - \left(\|\vt\|_\infty + \frac{1}{4} \|\vLambda^\top \vLambda\|_2 \right) \|\vpsi\|^2 \ge C\|\vpsi\|^2.
\end{align*}
In particular, $Q_0:=Q(\vH_0)=Q(-\lapl + \vV)$ is the form domain of $\vH_0$. It is dense and compactly embedded (see proof of \cref{fllexist}) in $\HC$, and it forms a Hilbert space itself with respect to the norm $\sqrt{ \|\vpsi\|^2 + \dua{\vpsi}{(\vH_0+C)\vpsi}}$.

From the discussion above we see that $\vH_0 = -\lapl  + \vV - \vt\cdot \vsigma_x$ would be another possible choice for the basic Hamiltonian that is almost equivalent and that is bounded below even in the limit $M\to\infty$.
Yet, we stick to the form \cref{h0def} that is linear in $\vLambda$ since this feature will be important in \cref{sec:AC} where the adiabatic connection is analyzed.

In this article, we consider the Hamiltonian $\vH_0$ with an additional linear coupling, of both matter and light parts, to external potentials $\vv\in\RR^N$ and $\vj\in\RR^M$ respectively, i.e.,
\begin{equation}\label{hvjdef}
\vH(\vv,\vj)=\vH_0 + \vv\cdot \vsigma_z + \vj\cdot \vx.
\end{equation}
The following virial result is of independent interest.

\begin{theorem}[Virial]\label{thm:first-virial}
For a ground state $\vpsi$ of $\vH(\vv,\vj)$ the relations
\begin{align*}
\|\grad\vpsi\|^2 &=  \|\vx\vpsi\|^2 + \frac{1}{2}\dua{\vpsi}{\vx\cdot\vLambda\vsigma_z\vpsi} +  \frac{1}{2} \vj\cdot\dua{\vpsi}{\vx\vpsi},\\
\dua{\vpsi}{\vt\cdot\vsigma_x\vpsi} &= -\frac{2}{M}\Re \dua{\vpsi}{ (\vt\cdot\vsigma_x)(\vx\cdot\grad)\vpsi}
\end{align*}
hold true.
\end{theorem}

\subsection{Constraints}\label{constrsec}
In what follows, we will often employ certain constraints on the wavefunction. For any $\vpsi\in\HC$, we define the \emph{magnetization}
$$
\vsigma_\vpsi=\dua{\vpsi}{\vsigma_z\vpsi}:=(\dua{\vpsi}{\sigma_z^1\vpsi},\ldots,\dua{\vpsi}{\sigma_z^N\vpsi})^\top \in \RR^N,
$$
where here and henceforth we employ the usual convention for ``vector-valued'' inner products. By the Cauchy--Schwarz inequality, we have that 
$\vsigma_\vpsi\in[-1,1]^N$ for any normalized $\vpsi\in\HC$.

Moreover, for any $\vpsi\in Q_0$ we define the \emph{displacement of $\vpsi$} as the vector
$$
\vxi_\vpsi=\dua{\vpsi}{\vx\vpsi}=\int_{\RR^M} \vx |\vpsi(\vx)|^2\,\dd\vx \in \RR^M,
$$
which is well-defined due to the Cauchy--Schwarz inequality (see also \eqref{eq:xi-finite}).

It will be instructive to explicitly spell out the relations $\|\vpsi\|=1$ and $\vsigma_\vpsi=\vsigma$ for the cases $N=1$ and $N=2$.

\begin{example}\label{ex:psi-vanish}
For $N=1$, we simply have
$$
\left.
\begin{aligned}
1&=\|\psi^+\|^2 + \|\psi^-\|^2\\
\sigma&=\|\psi^+\|^2 - \|\psi^-\|^2
\end{aligned}
\right\}
$$
so that $\|\psi^+\|^2=\frac{1+\sigma}{2}$ and $\|\psi^-\|^2=\frac{1-\sigma}{2}$. This immediately shows that if $\sigma=+1$, then $\psi^-\equiv 0$ 
and if $\sigma=-1$, then $\psi^+\equiv 0$. Moreover, these implications can be reversed, so that $\psi^+\not\equiv 0$ and $\psi^-\not\equiv 0$ precisely if $\sigma\in(-1,1)$.
Unfortunately, this is no longer true for $N\ge 2$.\\
\end{example}
\begin{example}\label{ex:ConstraintN2}
For $N=2$, we have
$$
\left.
\begin{aligned}
1&=\|\psi^{++}\|^2 + \|\psi^{+-}\|^2 + \|\psi^{-+}\|^2 + \|\psi^{--}\|^2\\
\sigma_1&=\|\psi^{++}\|^2 + \|\psi^{+-}\|^2 - \|\psi^{-+}\|^2 - \|\psi^{--}\|^2\\
\sigma_2&=\|\psi^{++}\|^2 - \|\psi^{+-}\|^2 + \|\psi^{-+}\|^2 - \|\psi^{--}\|^2
\end{aligned}
\right\}
$$
Adding and subtracting the last two equations from the first one, we obtain the following relations.
\begin{equation}\label{sigma2rel}
\left.
\begin{aligned}
\frac{1+\sigma_1}{2}&= \|\psi^{++}\|^2 + \|\psi^{+-}\|^2\\
\frac{1-\sigma_1}{2}&= \|\psi^{-+}\|^2 + \|\psi^{--}\|^2
\end{aligned}
\right\}
\quad
\left.
\begin{aligned}
\frac{1+\sigma_2}{2}&= \|\psi^{++}\|^2 + \|\psi^{-+}\|^2\\
\frac{1-\sigma_2}{2}&= \|\psi^{+-}\|^2 + \|\psi^{--}\|^2
\end{aligned}
\right\}
\end{equation}
From this it is apparent that whenever $\sigma_1=\pm 1$ or $\sigma_2=\pm 1$ (or both), a certain spinor component of $\vpsi$ must vanish.
But contrary to the $N=1$ case it is also possible that one (or more) spinor components of $\vpsi$ vanish even though $\vsigma\in(-1,1)^2$.

\begin{enumerate}[(i)]
\item\label{ex:ConstraintN2(i)} For instance, we have $\vsigma_\vpsi=\vsigma\in(-1,1)^2$ with $\sigma_1+\sigma_2<0$, if $\psi^{++}\equiv 0$ and $\|\psi^{+-}\|^2=\frac{1+\sigma_1}{2}>0$,
$\|\psi^{-+}\|^2=\frac{1+\sigma_2}{2}>0$ and $\|\psi^{--}\|^2= - \frac{\sigma_1+\sigma_2}{2}>0$.
\item We can even have $\psi^{+-}=\psi^{-+}\equiv 0$
if $\sigma=\sigma_1=\sigma_2\in (-1,1)$ by taking $\|\psi^{++}\|^2=\frac{1+\sigma}{2}$ and $\|\psi^{--}\|^2=\frac{1-\sigma}{2}$.
Similarly, $\psi^{++}=\psi^{--}\equiv 0$ if $\sigma=\sigma_1=-\sigma_2\in (-1,1)$, by taking $\|\psi^{+-}\|^2=\frac{1+\sigma}{2}$ and $\|\psi^{-+}\|^2=\frac{1-\sigma}{2}$.
\item However, three spinor components can only vanish for a $\vpsi\in\HC$ with $\vsigma_\vpsi=\vsigma=(\pm 1,\pm 1)$.
\end{enumerate}
\end{example}

In summary, if $N\ge 2$, $\vsigma\in(-1,1)^N$ does \emph{not} imply that $\psi^\valpha\not\equiv 0$ for all $\valpha$.

\section{Main results}
\label{sec:main}
In this section, we present our main results. 
The proofs are deferred until~\cref{sec:proofs}.

\subsection{Hohenberg--Kohn theorem}\label{sec:HK}

We begin our discussion with a Hohenberg--Kohn-type theorem. In order to state this, we need a definition which turns out to be crucial for the rest of the article.

Let the $N\times 2^N$ matrix $\vOmega$ be given by $\Omega_{n,\valpha}=(\sigma_z^n)_{\valpha\valpha}$, i.e., the matrix with 
the diagonal of $\sigma_z^n$ as the $n$-th row vector. 
We say that $\vsigma\in[-1,1]^N$ is \emph{regular} if for every $\vomega\in\RR^{2^N}$ with $\omega_\valpha\ge 0$ and $\sum_{\valpha} \omega_\valpha=1$ that verifies $\vOmega\vomega=\vsigma$, we have $\Aff \{\vOmega\ve_\valpha : \omega_\valpha\neq 0\}=\RR^N$, where $\Aff$ denotes the affine hull and $\ve_\valpha$ denotes the standard basis vector of $\RR^{2^N}$.
We denote the set of regular $\vsigma$'s by $\RC_N$.

\begin{example}
If $N=1$ then $\sigma\in[-1,1]$ is regular iff $\sigma\in(-1,1)$. In fact, $S=\{\vOmega\ve_\valpha : \omega_\valpha\neq 0\}\subset\{-1,1\}$
and so $\Aff(S)= \RR$ iff $|S|=2$. But $\vOmega\vomega=\sigma$ simply reads
$\omega_+-\omega_-=\sigma$, and $\omega_+, \omega_-\neq0$, $\omega_+ \neq \omega_-$ iff $\sigma\neq \pm 1$.\\
\end{example}
\begin{example}
When $N=2$, $S\subset\{ (\pm 1,\pm 1) \}$, i.e., the vertices of the unit square. Then $\Aff(S)=\RR^2$ holds iff $|S|\ge 3$.
According to the discussion after \cref{sigma2rel}, in order for $|S|\ge 3$ to hold, it is necessary and sufficient that $\vsigma\in\RC_2$, where
$$
\RC_2=\{ (\sigma_1,\sigma_2)\in(-1,1)^2 : \sigma_1\neq\sigma_2, \; \sigma_1\neq -\sigma_2 \}.
$$
In light of \cref{ex:ConstraintN2}, case \eqref{ex:ConstraintN2(i)}, a spinor component can vanish also for $\vsigma \in \RC_2$.
\end{example}

Note that the regular set breaks up into disjoint components. 
\begin{proposition} \label{prop:RN}
Let $N\ge 1$. Then $\RC_N$ is the union of disjoint open convex polytopes. Also, $[-1,1]^N\setminus\RC_N$ is the union of a finite number of hyperplanes intersected with $[-1,1]^N$.
\end{proposition}

See \cref{fig:RegSetEx} for a sketch of the regular sets $\RC_2$ and $\RC_3$.
The importance of regular $\vsigma$'s is explained by the following theorem.

 \begin{figure}
     \begin{subfigure}[b]{0.28\textwidth}
          \centering
          \resizebox{\linewidth}{!}{\begin{tikzpicture}[every node/.style={inner sep=0, outer sep=0pt}, line width=.7pt]
        
    \coordinate (a) at (-2,-2);
    \coordinate (b) at (2,-2);
    \coordinate (c) at (2,2);
    \coordinate (d) at (-2,2);
    
    \draw (a) -- (b) -- (c) -- (d) -- (a);
    \draw (a) -- (c);
    \draw (b) -- (d);


\end{tikzpicture}}  
     \end{subfigure}
     \begin{subfigure}[b]{0.33\textwidth}
          \centering
          \resizebox{\linewidth}{!}{\tdplotsetmaincoords{80}{120}
\usetikzlibrary{3d}

\begin{tikzpicture}[scale=2,tdplot_main_coords]
\tikzstyle{grid}=[thin,color=red,tdplot_rotated_coords]

  
  \draw [thick] (0,0,0) -- (2,0,0);
  \draw [thick] (0,0,0) -- (0,2,0);
  \draw [thick] (2,2,0) -- (2,0,0);
  \draw [thick] (2,2,0) -- (0,2,0);
  
  \draw [thick] (0,0,2) -- (2,0,2);
  \draw [thick] (0,0,2) -- (0,2,2);
  \draw [thick] (2,2,2) -- (2,0,2);
  \draw [thick] (2,2,2) -- (0,2,2);
  
  \draw [thick] (0,0,2) -- (0,0,0);
  \draw [thick] (2,0,2) -- (2,0,0);
  \draw [thick] (0,2,2) -- (0,2,0);
  \draw [thick] (2,2,2) -- (2,2,0);

  \draw [dashed] (0,0,0) -- (2,2,2);
  \draw [dashed] (2,0,0) -- (0,2,2);
  \draw [dashed] (0,2,0) -- (2,0,2);
  \draw [dashed] (0,0,2) -- (2,2,0);
  
  \draw [dashed] (0,0,0) -- (2,0,2);
  \draw [dashed] (2,0,0) -- (0,0,2);
  
  \draw [dashed] (0,0,0) -- (0,2,2);
  \draw [dashed] (0,2,0) -- (0,0,2);
  
  \draw [dashed] (0,0,0) -- (2,2,0);
  \draw [dashed] (0,2,0) -- (2,0,0);
  
  \draw (2,0,0) -- (2,2,2);
  \draw (2,2,0) -- (2,0,2);
  
  \draw (0,2,0) -- (2,2,2);
  \draw (2,2,0) -- (0,2,2);
  
  \draw (0,0,2) -- (2,2,2);
  \draw (2,0,2) -- (0,2,2);
  
  \draw [dashed] (0,1,1) -- (2,1,1);
  \draw [dashed] (1,0,1) -- (1,2,1);
  \draw [dashed] (1,1,0) -- (1,1,2);
  
  



\end{tikzpicture}}  
     \end{subfigure}
     \caption{Left: The set $\RC_2\subset(-1,1)^2$ is the union of 4 congruent open triangles. Right: The set $\RC_3\subset(-1,1)^3$ is the union of 24 congruent open tetrahedra.}
     \label{fig:RegSetEx}
 \end{figure}
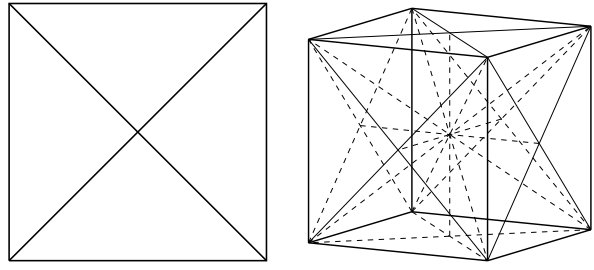

\begin{theorem}[Hohenberg--Kohn]\label{hkthm}
Fix $\vsigma\in[-1,1]^N$ and $\vxi\in\RR^M$.
Let $\vv^{(1)},\vv^{(2)}\in \RR^{N}$ and $\vj^{(1)},\vj^{(2)}\in\RR^M$, 
and suppose that $\vpsi^{(1)},\vpsi^{(2)}\in Q_0$ are ground states of $\vH(\vv^{(1)},\vj^{(1)})$ and $\vH(\vv^{(2)},\vj^{(2)})$ 
respectively.\\
If $\vsigma=\vsigma_{\vpsi^{(1)}}=\vsigma_{\vpsi^{(2)}}$ and $\vxi=\vxi_{\vpsi^{(1)}}=\vxi_{\vpsi^{(2)}}$, then $\vpsi^{(1)}$ is also a ground state of $\vH(\vv^{(2)},\vj^{(2)})$ and $\vpsi^{(2)}$ is also a ground state of $\vH(\vv^{(1)},\vj^{(1)})$.
Furthermore, $\vj=\vj^{(1)}=\vj^{(2)}$ and 
\begin{enumerate}[(i)]
\item\label{hkthm:item:regular} (Regular case) If $\vsigma$ is regular, then $\vv^{(1)}=\vv^{(2)}$.
\item\label{hkthm:item:irregular} (Irregular case) Otherwise, for all $\valpha\in I^{(1)}\cup I^{(2)}$ there holds
\begin{equation*}
\sum_{n=1}^N (\vsigma_z^n)_{\valpha\valpha}(v^{(1)}_n-v^{(2)}_n) = E(\vv^{(1)},\vj) -  E(\vv^{(2)},\vj),
\end{equation*}
where $I^{(i)}$ denotes the set of spinor indices $\valpha$ for which $(\psi^{(i)})^\valpha\not\equiv 0$.
\end{enumerate}
\end{theorem}

The regularity property of the magnetization vector $\vsigma$ can be seen in analogy to the condition on the zeros of the wave function in finite-lattice DFT \cite[Cor.~10]{penz2021-Graph-DFT}. 
If indeed $\vv^{(1)}\neq \vv^{(2)}$ in the irregular case, then by a similar argument as in ~\cite[Th.~9]{penz2023geometry} any convex combination of $\vv^{(1)}, \vv^{(2)}$ also has the same $\vsigma$ as the ground-state magnetization.

Notice that unlike the Hohenberg--Kohn theorem for the electronic Hamiltonian, the potentials are completely determined in the regular case, i.e., not only up to an additive constant.
The theorem itself is nonconstructive in nature, more precisely it only states the injectivity of the ``potential to ground-state density map'' $(\vv,\vj)\mapsto (\vsigma,\vxi)$
and \emph{not} its surjectivity. Whenever $(\vsigma,\vxi)\in[-1,1]^N\times \RR^M$ corresponds to a ground state of $\vH(\vv,\vj)$ for some $(\vv,\vj)\in\RR^N\times\RR^M$,
then we say that $(\vsigma,\vxi)$ is \emph{$v$-representable}.
Since a ground state can either be an element of the Hilbert space (pure state) or have the form of a statistical mixture expressed by a density matrix acting on the Hilbert space (ensemble state), we respectively speak about \emph{pure-state $v$-representability} and \emph{ensemble $v$-representability}.
These are not to be confused with the $N$-representability concept below.

\subsection{Levy--Lieb functional}\label{llsec}
The preceding discussion suggests that we consider functionals of the ``density'' pair $(\vsigma,\vxi)$.
The objective is then to formulate the ground-state problem in terms of $(\vsigma,\vxi)$ only. Following the standard DFT recipe, as a first step we minimize the internal energy $\vpsi \mapsto \dua{\vpsi}{\vH_0\vpsi}$ under the constraints $\vsigma_\vpsi = \vsigma$ and $ \vxi_\vpsi = \vxi$.
This gives rise to the Levy--Lieb functional, also commonly called the pure-state constrained-search functional ~\cite{Levy79,Lieb1983}. In order for this functional to be well-defined, we must first show that to any $(\vsigma,\vxi)$
there corresponds at least one wavefunction. If this is the case, we call $(\vsigma,\vxi)$ \emph{$N$-representable}. We caution the Reader that this is standard terminology in DFT (and its variants), where $N$ has nothing to do with the number of two-level systems considered here.
Every $\vpsi\in Q_0$ with $\|\vpsi\|=1$ has $(\vsigma_\vpsi,\vxi_\vpsi)\in[-1,1]^N\times\RR^M$.
Fortunately, it is simple to show also the converse, that every $(\vsigma,\vxi)\in[-1,1]^N\times\RR^M$ is $N$-representable.

\begin{theorem}[$N$-representability]\label{repthm}
For every $(\vsigma,\vxi)\in[-1,1]^N\times\RR^M$ there exists $\vpsi\in Q_0$ such that $\|\vpsi\|=1$, $\vsigma_\vpsi=\vsigma$ and $\vxi_\vpsi=\vxi$.
\end{theorem}

We introduce the constraint manifold that collects all states that map to a given $(\vsigma,\vxi)\in[-1,1]^N\times\RR^M$,
\begin{equation*}
\MC_{\vsigma,\vxi}=\left\{ \vpsi\in Q_0 : \|\vpsi\|=1, \; \vsigma_\vpsi=\vsigma, \; \vxi_\psi=\vxi \right\}.
\end{equation*}
Using the preceding theorem, we may write for any $(\vv,\vj)\in\RR^N\times\RR^M$ that
\begin{equation}\label{fllderiv}
\begin{aligned}
E(\vv,\vj)&=\inf_{\substack{\vpsi\in Q_0 \\ \|\vpsi\|=1}} \dua{\vpsi}{\vH(\vv,\vj)\vpsi}\\
&=\inf_{(\vsigma,\vxi)\in[-1,1]^N\times\RR^M} 
\Bigg[ \inf_{\vpsi\in\MC_{\vsigma,\vxi}} \dua{\vpsi}{\vH(\vv,\vj)\vpsi} \Bigg]\\
&=\inf_{(\vsigma,\vxi)\in[-1,1]^N\times\RR^M} 
\Bigg[ \inf_{\vpsi\in\MC_{\vsigma,\vxi}} \dua{\vpsi}{\vH_0\vpsi}  + \dua{\vpsi}{\vv\cdot\vsigma_z\vpsi} + \dua{\vpsi}{\vj\cdot\vx\vpsi}\Bigg]\\
&= \inf_{(\vsigma,\vxi)\in[-1,1]^N\times\RR^M} 
\Big[ F_\LL(\vsigma,\vxi) + \vv\cdot \vsigma + \vj\cdot \vxi \Big],
\end{aligned}
\end{equation}
where we used \cref{hvjdef} and we defined the \emph{Levy--Lieb (universal density) functional} $F_\LL : [-1,1]^N\times \RR^M \to \RR$ via
\begin{equation*}
F_\LL(\vsigma,\vxi) = \inf_{\vpsi\in\MC_{\vsigma,\vxi}} \dua{\vpsi}{\vH_0\vpsi}
\end{equation*}
for every $(\vsigma,\vxi) \in [-1,1]^N\times \RR^M$.
Clearly, $|F_\LL(\vsigma,\vxi)|<+\infty$.
An immediate question is whether the ``$\inf$'' is attained in the definition of $F_\LL$.
\begin{theorem}[Existence of an optimizer for $F_\LL$] \label{fllexist}
For every $(\vsigma,\vxi)\in [-1,1]^N\times \RR^M$
there exists a $\vpsi\in\MC_{\vsigma,\vxi}$ such that $F_\LL(\vsigma,\vxi) = \dua{\vpsi}{\vH_0\vpsi}$.
\end{theorem}
The proof of this result is somewhat different from the analogous one in standard DFT~\cite{Lieb1983} or in generalization to paramagnetic current-DFT~\cite{Laestadius2014,Kvaal2021}: there, one exploits the density constraint on the
wavefunction to obtain the tightness of the optimizing sequence. In our case, the trapping nature of $\vH_0$ provides compactness.
Using the preceding result, we may employ trial state constructions to derive useful properties of $F_\LL$. 

\begin{theorem}[Properties of $F_\LL$]\label{fllprop}
For every $(\vsigma,\vxi)\in[-1,1]^N\times\RR^M$ the following hold true.
\begin{enumerate}[(i)]
\item\label{fllprop:item:displacement-rule} (Displacement rule) For any $\vzeta\in\RR^M$ the formula
$$
F_\LL(\vsigma,\vxi+\vzeta)=F_\LL(\vsigma,\vxi) + 2\vzeta\cdot\vxi + \vzeta \cdot \vLambda\vsigma  +  |\vzeta|^2 
$$
holds. In particular,
$$
F_\LL(\vsigma,\vxi)=F_\LL(\vsigma,\bm{0}) + \vxi \cdot \vLambda\vsigma  +  |\vxi|^2.
$$
\item\label{fllprop:item:real} There is a real-valued optimizer of $F_\LL(\vsigma,\vxi)$.
\item\label{fllprop:item:virial} (Virial relation) For any optimizer $\vpsi$ of $F_\LL(\vsigma,\bm{0})$ the formula
$$
\|\grad\vpsi\|^2 - \|\vx\vpsi\|^2  = \frac{1}{2} \dua{\vpsi}{\vx\cdot \vLambda\vsigma_z\vpsi}
$$
holds true.
\item\label{fllprop:item:coupling} For a real-valued optimizer $\vpsi$ of $F_\LL(\vsigma,\vxi)$, the formula
\begin{align*}
 \dua{\vpsi}{\vx\cdot \vLambda\vsigma_z\vpsi} =&\; \vxi\cdot\vLambda\vsigma - \frac{1}{2} \|\vLambda\vsigma_z\vpsi\|^2 + \frac{1}{2}|\vLambda\vsigma|^2 - \dua{\vt\cdot\vsigma_x\vpsi}{\grad\cdot\vLambda(\vsigma_z-\vsigma)\vpsi}
\end{align*}
holds true.
\end{enumerate}
\end{theorem}
It readily follows from \eqref{fllprop:item:displacement-rule} that the function $\vxi\mapsto F_\LL(\vsigma,\vxi)$ is smooth and convex for every fixed $\vsigma\in[-1,1]^N$.

Next, we consider the constrained minimization problem defining $F_\LL(\vsigma,\vxi)$ from a geometric perspective.

\begin{lemma}\label{tangthm}
Let $(\vsigma,\vxi)\in\RC_N\times\RR^M$.
Then the following statements hold true.
\begin{enumerate}[(i)]
\item\label{tangthm:item:submersed} $\MC_{\vsigma,\vxi}$ is a closed submersed Hilbert submanifold of $Q_0$.
\item\label{tangthm:item:tangent-space} The tangent space of $\MC_{\vsigma,\vxi}$ at $\vpsi\in\MC_{\vsigma,\vxi}$ is given by
\begin{align*}
\TC_\vpsi(\MC_{\vsigma,\vxi})&=\Big\{ \vchi\in Q_0 : \dua{\vpsi}{\vchi}=0,\; \dua{\vsigma_z\vpsi}{\vchi}=0, \; 
\dua{\vx\vpsi}{\vchi}=0 \Big\},
\end{align*}
which we consider as a vector space over $\RR$.\footnote{This result for the tangent space corresponds to wavefunctions that are taken modulo a global phase, thus actually the complex projective space of $\mathcal{H}$ is considered here.}\item\label{tangthm:item:cotangent-space} 
The orthogonal complement of $\TC_\vpsi(\MC_{\vsigma,\vxi})$ in $\HC$ is the $(N+M+1)$-dimensional vector space given by
\begin{align*}
\TC_\vpsi(\MC_{\vsigma,\vxi})^\perp=\Big\{ \left(E + \vv\cdot\vsigma_z + \vj\cdot \vx\right)\vpsi \in \HC &:  E\in\RR, \; \vv\in\RR^{N}, \; \vj\in\RR^M \Big\}.
\end{align*}
\end{enumerate}
\end{lemma}

The Lagrange multiplier rule and the positivity of the Hessian give the following straightforward result. Note that a similar calculation does not seem to be possible in the setting of standard DFT, since there the density constraint does not give rise to a well-defined tangent space.

\begin{theorem}[Optimality]\label{elthm}
Let $(\vsigma,\vxi)\in \RC_N\times\RR^M$ and suppose that $\vpsi\in\MC_{\vsigma,\vxi}$ is an optimizer of $F_\LL(\vsigma,\vxi)$.
Then there exist Lagrange multipliers $E\in\RR$, $\vv\in\RR^N$ and $\vj\in\RR^M$, such that $\vpsi$ satisfies the strong Schr\"odinger equation
\begin{equation}\label{elschr}
\vH(\vv,\vj) \vpsi=E\vpsi
\end{equation}
and the second-order condition
\begin{equation}\label{seccond}
\dua{\vchi}{\vH(\vv,\vj)\vchi}\ge E\|\vchi\|^2
\end{equation}
for all $\vchi\in\TC_\vpsi(\MC_{\vsigma,\vxi})$.
Moreover, 
\begin{equation*}
F_\LL(\vsigma,\vxi)=\dua{\vpsi}{\vH_0\vpsi}=E - \vv\cdot\vsigma - \vj\cdot\vxi.
\end{equation*}
\end{theorem}

It is also possible to write down the optimality conditions if $\vsigma$ is irregular, but we do not consider that case in detail here. Yet, we will do so later in \cref{thrm:EulerLagrange} for a reduced model.
Note that the above theorem says that optimizers of the constrained-search functional are solutions of the Schrödinger equation, yet it does not guarantee that they are ground states. Theorem~\ref{aufbauthm} below shows that they are at least low-lying eigenstates. But before that, we state the following characterization for ground-state optimizers in terms of degeneracy.

\begin{theorem}[Hohenberg--Kohn-type result for optimizers]\label{hkopt}
Let $(\vsigma,\vxi)\in \RC_N\times\RR^M$. Suppose that an optimizer of $F_\LL(\vsigma,\vxi)$ is a ground state of $\vH(\vv,\vj)$
for some $E=E(\vv,\vj)\in\RR$, $\vv\in\RR^N$ and $\vj\in\RR^M$. Then all other optimizers of $F_\LL(\vsigma,\vxi)$ which are ground states (for possibly different potentials and
energies) must also be in $\ker(\vH(\vv,\vj)-E(\vv,\vj))$.
\end{theorem}

In other words, there is no real ``competition'' for the optimizers which happen to be ground states: they all belong to the same degenerate eigenspace of the same Hamiltonian. We immediately see that if this eigenspace happens to be one-dimensional then there can only be one optimizer which is a ground state. This will be the case for $N=M=1$, see \cref{specsec} below. 

The second-order information \cref{seccond} about a minimizer gives a result which is analogous to the Aufbau principle in Hartree--Fock theory.
\begin{theorem}[Optimizers are low-lying eigenstates]\label{aufbauthm}
Let $(\vsigma,\vxi)\in \RC_N\times\RR^M$, and suppose that $\vpsi\in\MC_{\vsigma,\vxi}$ is an optimizer of $F_\LL(\vsigma,\vxi)$, with
Lagrange multipliers $E\in\RR$, $\vv\in\RR^N$ and $\vj\in\RR^M$, so that 
\cref{elschr} and \cref{seccond} hold true.
Then $\vpsi$ is at most the $(N+M)$th excited eigenstate of $\vH(\vv,\vj)$. 
\end{theorem}

With this result we can conclude that any $(\vsigma,\vxi)\in \RC_N\times\RR^M$, while not proven to be pure-state $v$-representable in the usual sense, can be called ``low-lying excited-state $v$-representable''. In \cref{flsec} below we additionally prove ensemble $v$-representability.

\subsection{Adiabatic connection}
\label{sec:AC}

Next, we consider a useful tool in DFT \cite{levy1991-AC}: the adiabatic connection from zero to full coupling, which is essentially based on the Newton--Leibniz formula applied to the density functional as a function of the coupling strength, at fixed density.
In order to formulate the adiabatic connections in the current setting, some preparations are due.

In this section, we will indicate the dependence on $\vLambda$ in the Levy--Lieb functional by a superscript. It is easy to see that $\RR^{M\times N}\ni \vLambda\mapsto F_\LL^\vLambda(\vsigma,\vxi)$ and thus in particular $\RR\ni s\mapsto F_\LL^{s\vLambda}(\vsigma,\vxi)$ are concave for every fixed $(\vsigma,\vxi)\in[-1,1]^N\times\RR^M$.
The concave equivalent of the subdifferential, commonly called superdifferential, of this functional is easy to determine. We
remind the Reader that if $X$ is a vector space and $f:X\to\RR$ concave, the superdifferential of $f$ is a set-valued mapping $\ol{\partial}f : X \to \mathcal{P}(X^*)$
given by
$$
\ol{\partial} f(x) = \{ \phi \in X^* : f(x')-f(x)\le \dua{\phi}{x'-x}\;\text{for all}\; x'\in X\}
$$
for any $x\in X$. 

The functional at coupling strength $\vLambda$ can be given by the generalized Newton--Leibniz formula,
\begin{equation}\label{fll-newtonleibniz}
F_\LL^{\vLambda}(\vsigma,\vxi)=F_\LL^{0}(\vsigma,\vxi)+\int_0^1 \ol{\partial}_s F_\LL^{s\vLambda}(\vsigma,\vxi)(s) \,\dd s,
\end{equation}
where the integral is independent of the choice of elements from $\ol{\partial}_s F_\LL^{s\vLambda}(\vsigma,\vxi)(s)$. Here, the value of $F_\LL^\vLambda$ at zero coupling is given as follows.

\begin{lemma}\label{lemma:FLL0}
Let $(\vsigma,\vxi) \in [-1,1]^N\times \RR^M$ and $\vLambda = \bm{0}$, then the wavefunction  
\begin{equation} \label{eq:FLL0TrialState}
\vpsi(\vx) = \pi^{-M/4}e^{-|\vx - \vxi|^2/2}\mathbf{c}, \quad  \mathbf{c} = 2^{-N/2}\bigotimes_{n=1}^N \begin{pmatrix}\sqrt{1+\sigma_n} \\ \sqrt{1-\sigma_n} \end{pmatrix}, 
\end{equation}
is an optimizer of $F^0_\mathrm{LL}(\vsigma,\vxi)$, and 
\begin{equation*}
 F^0_\mathrm{LL}(\vsigma,\vxi) = M + |\vxi|^2 - \sum_{n=1}^N t_n \sqrt{1-\sigma_n^2}.
\end{equation*}
\end{lemma}

Moreover, the superdifferential in the integrand of the Newton--Leibniz formula verifies the following chain rule.
\begin{lemma}\label{lemma-FLL-subdiff}
Fix $(\vsigma,\vxi)\in[-1,1]^N\times \RR^M$ and $s\in\RR$. Then
\begin{align*}
\ol{\partial}_s F_\LL^{s\vLambda}(\vsigma,\vxi)(s)\supset \{ \dua{&\vpsi_{s\vLambda}}{\vx\cdot\vLambda\vsigma_z\vpsi_{s\vLambda}} : \\
&\vpsi_{s\vLambda}\in Q_0 \;\;\text{with}\;\; F_\LL^{s\vLambda}(\vsigma,\vxi)=\dua{\vpsi_{s\vLambda}}{\vH_0^{s\vLambda}\vpsi_{s\vLambda}} \}.
\end{align*}
\end{lemma}

Using \emph{first} the displacement rule, \cref{fllprop}~\eqref{fllprop:item:displacement-rule}, then the Newton--Leibniz rule and \cref{lemma-FLL-subdiff} (with $\vxi = \bm{0}$) gives
\begin{align*} 
    F_\LL^{\vLambda}(\vsigma,\vxi) 
    = \vxi\cdot \vLambda\vsigma + |\vxi|^2  + F_\LL^{0}(\vsigma,\bm{0}) + \int_0^1 \dua{\vpsi_{s\vLambda}}{\vx\cdot\vLambda\vsigma_z\vpsi_{s\vLambda}} \,\dd s,
\end{align*}
where, following \cref{fllprop}~\eqref{fllprop:item:real}, $\vpsi_{s\vLambda} \in \MC_{\vsigma,\bm{0}}$ can always be chosen as a \emph{real-valued} optimizer for $F_\LL^{s\vLambda}(\vsigma,\bm{0})$. This further allows the use of \cref{fllprop}~\eqref{fllprop:item:coupling} for the integrand and we may also employ the result from \cref{lemma:FLL0} for $F_\LL^{0}(\vsigma,\bm{0})$. In summary, we obtain the adiabatic connection for the Levy--Lieb functional:

\begin{theorem}[Adiabatic connection for $F_\LL$]\label{thm:FLLadia} 
The functional $F_\LL^{\vLambda}:[-1,1]^N \times \RR^M \to \RR$ satisfies
    \begin{align*}
    F^{\vLambda}_\mathrm{{LL}}(\vsigma, \vxi) &= M  + |\vxi|^2 - \sum_{n=1} t_n \sqrt{1-\sigma_n^2} + \vxi\cdot \vLambda\vsigma +  G^\vLambda(\vsigma) \quad\text{with} \\
    G^\vLambda(\vsigma) &= \frac{1}{2}|\vLambda\vsigma|^2 -\int_0^1 \left(\frac{1}{2} \|\vLambda\vsigma_z\vpsi_{s\vLambda}\|^2  + \dua{\vt\cdot\vsigma_x\vpsi_{s\vLambda}}{\grad\cdot \vLambda(\vsigma_z-\vsigma)\vpsi_{s\vLambda}}\right)\dd{s} ,
\end{align*}
where $\vpsi_{s\vLambda} \in \MC_{\vsigma,\bm{0}}$ is a real-valued optimizer for $F_\LL^{s\vLambda}(\vsigma,\bm{0})$.
\end{theorem}

Note here that $G^\vLambda(\vsigma)$ is independent of $\vxi$. It is defined in analogy to the exchange-correlation functional of standard DFT as the difference between the interacting and the non-interacting ($\vLambda =\bm{0}$) Levy--Lieb functionals, minus the direct coupling term $\vxi\cdot \vLambda\vsigma$.

To conclude the pure-state formulation of QEDFT, we repeat the ground-state energy at coupling strength $\vLambda$,
\begin{equation*}
 E^{\vLambda}(\vv,\vj) =  \inf_{(\vsigma,\vxi)\in[-1,1]^N\times\RR^M} 
\Big[ F_\LL^{\vLambda}(\vsigma,\vxi) + \vv\cdot \vsigma + \vj\cdot \vxi \Big],
\end{equation*}
where $F_\LL^{\vLambda}$ can be determined from~\cref{thm:FLLadia}. The critical, unknown term is $G^\vLambda(\vsigma)$, which is the functional that all approximations in DFT aim at.

\subsection{Lieb functional}\label{flsec}

In contrast to the Levy--Lieb functional, the Lieb functional~\cite{Lieb1983} is the constrained-search over \emph{mixed} states, represented by density matrices.
To us, a density matrix is a self-adjoint, positive and trace-class operator $\vGamma : \HC\to\HC$, normalized to unit trace. We denote 
the integral kernel of $\vGamma$ with the same symbol: $\vGamma(\vx,\vx')$ which is a square-integrable $\RR^M\times\RR^M\to \CC^{2^N\times 2^N}$ function.
Moreover, we define
$$
\vsigma_\vGamma= \Tr\vsigma_z\vGamma=(\Tr \sigma_z^1\vGamma,\ldots, \Tr \sigma_z^N\vGamma)\in[-1,1]^N,
$$
and with the spin-summed density matrix $\sps{\Gamma}=\Tr_{\CC^{2^N}} \vGamma$ we set
$$
\vxi_\vGamma=\int_{\RR^M} \vx \sps{\Gamma}(\vx,\vx)\,\dd\vx.
$$
For notational convenience we introduce the following subset of density matrices,
$$
\DC(Q_0)=\{ \vGamma\in\mathfrak{S}_1(Q_0) : 0\le\vGamma=\vGamma^\dag, \, \Tr\vGamma=1,\, \vH_0\vGamma\in \mathfrak{S}_1(\HC) \},
$$
or, more explicitly,
\begin{equation}\label{eq:Gamma-set}
\begin{aligned}
\DC(Q_0)=\Bigg\{ \vGamma\in\mathfrak{S}_1(Q_0) &: \vGamma=\sum_{j=1}^\infty c_j \dyad{\vpsi_j}, \; c_j\in\RR_+, \;\sum_{j=1}^\infty c_j=1,\\
&\vpsi_j\in Q_0\;\text{are}\;L^2\text{-orthonormal},\\
&\sum_{j=1}^\infty c_j \left(\|\grad\vpsi_j\|^2 + \|\vx\vpsi_j\|^2\right) < \infty \Bigg\}.
\end{aligned}
\end{equation}
Here, $\mathfrak{S}_p$ denotes the $p$-Schatten class and we denote by $\mathfrak{S}_\infty$ the compact operators.
Clearly, for $\vGamma\in\DC(Q_0)$ the quantity $\vxi_\vGamma$ is finite since
\begin{equation}\label{eq:xi-finite}
|\vxi_\vGamma|\le \sum_{j=1}^\infty c_j \int_{\RR^M} |\vx||\vpsi_j(\vx)|^2\, \dd\vx \le \frac{1}{2}\sum_{j=1}^\infty c_j \left(\|\vpsi_j\|^2 +  \|\vx\vpsi_j\|^2\right)<\infty,
\end{equation}
by the Cauchy--Schwarz inequality.

The calculation \cref{fllderiv} can be repeated using mixed states as well,
\begin{align}
E(\vv,\vj)&=\inf_{\vGamma\in\DC(Q_0)} \Tr(\vH(\vv,\vj)\vGamma) \label{eq:FLprop:Gamma-inf}\\
&=\inf_{(\vsigma,\vxi)\in[-1,1]^N\times\RR^M} 
\Bigg[ \inf_{\substack{\vGamma \in\DC(Q_0) \nonumber\\ \vsigma_\vGamma = \vsigma\\ \vxi_\vGamma=\vxi }} \Tr(\vH(\vv,\vj)\vGamma) \Bigg]\\
&=\inf_{(\vsigma,\vxi)\in[-1,1]^N\times\RR^M} 
\Bigg[ \inf_{\substack{\vGamma \in \DC(Q_0) \nonumber\\ \vsigma_\vGamma = \vsigma\\ \vxi_\vGamma=\vxi }} \Tr(\vH_0\vGamma)  + \Tr(\vv\cdot\vsigma_z\vGamma) 
+ \Tr(\vj\cdot\vx\vGamma)\Bigg] \nonumber\\
&= \inf_{(\vsigma,\vxi)\in[-1,1]^N\times\RR^M} 
\Big[ F_\LI(\vsigma,\vxi) + \vv\cdot \vsigma + \vj\cdot \vxi \Big].
\label{eq:FLprop:E-inf}
\end{align}
Here, we introduced the \emph{Lieb (universal density) functional} $F_\LI : \RR^N\times\RR^M\to\RR$ via
$$
F_\LI(\vsigma,\vxi)=\inf_{\substack{\vGamma\in\DC(Q_0)\\ \vsigma_\vGamma=\vsigma\\ \vxi_\vGamma=\vxi}} \Tr(\vH_0\vGamma)
$$
for all $\vsigma\in[-1,1]^N$ and $\vxi\in\RR^M$, and $F_\LI(\vsigma,\vxi)\equiv +\infty$ otherwise. 
Since $\vGamma \mapsto (\vsigma_\vGamma, \vxi_\vGamma)$ is linear, it is immediate from the definition by an infimum that $F_\LI$ is convex. 
An optimizer $\vGamma$ in \eqref{eq:FLprop:Gamma-inf} would be called a ground-state ensemble for the Hamiltonian $\vH(\vv,\vj)$.
By choosing $\vGamma=\dyad{\vpsi} $ 
in the preceding infimum, where $\vpsi$ is an optimizer for $F_\LL(\vsigma,\vxi)$, we obtain
$$
F_\LI(\vsigma,\vxi)\le F_\LL(\vsigma,\vxi),
$$
so that $F_\LI(\vsigma,\vxi)<+\infty$ for $(\vsigma,\vxi)\in[-1,1]^N\times\RR^M$ from the result before. Moreover, for such $(\vsigma,\vxi)$ we also have
$-\infty < F_\LI(\vsigma,\vxi)$ as the following basic result shows. 

\begin{theorem}[Existence of an optimizer for $F_\LI$]\label{flexsist}
For every $(\vsigma,\vxi)\in [-1,1]^N\times \RR^M$
there exists $\vGamma\in\DC(Q_0)$ such that $\vsigma_\vGamma=\vsigma$, $\vxi_\vGamma=\vxi$ and  $F_\LI(\vsigma,\vxi) = \Tr(\vH_0\vGamma)$.
\end{theorem}

In the next theorem we collect the general convex-analytic properties of $F_\LI$ that carry over from the standard DFT setting to our context.
We use some well-known results from convex analysis, see e.g.\ \cite{niculescu2006convex}.

\begin{theorem}[Convex-analytic properties of $F_\LI$]\label{flconvprop}
For the Lieb functional $F_\LI: \RR^N\times\RR^M\to\RR$, the following properties hold true.
\begin{enumerate}[(i)]
\item\label{flconvprop:item:lsc} $F_\LI$ is lower semicontinuous, i.e., if $\vsigma_j\to\vsigma$ in $\RR^N$ and $\vxi_j\to\vxi$ in $\RR^M$, then $\liminf_{j\to\infty} F_\LI(\vsigma_j,\vxi_j)\ge F_\LI(\vsigma,\vxi)$.
\item\label{flconvprop:item:convex} $F_\LI$ is the convex envelope of $F_\LL$ and as such $F_\LI\le F_\LL$. Moreover, $F_\LI$ is locally Lipschitz and hence a.e.\ differentiable in $(-1,1)^N\times\RR^M$. 
\item\label{flconvprop:item:subdiff} The subdifferential of $F_\LI$ reads
\begin{align*}
\underline{\partial} F_\LI(\vsigma,\vxi)=\{ & (-\vv,-\vj)\in\RR^{N}\times\RR^{M} : \Tr (\vH(\vv,\vj)\vGamma)=E(\vv,\vj)\\
&\text{for some}\;\vGamma\in\DC(Q_0)\;\text{with}\;\vsigma_\vGamma=\vsigma,\;\vxi_\vGamma=\vxi\}.
\end{align*}
We have that $\underline{\partial} F_\LI(\vsigma,\vxi)\neq \emptyset$ for all $(\vsigma,\vxi)\in(-1,1)^N\times\RR^M$.
\item\label{flconvprop:item:Legendre} $F_\LI$ is the Legendre transform of $E$,
\begin{equation*}
    F_\LI(\vsigma,\vxi) = \sup_{(\vv,\vj)\in \mathbb R^N \times \mathbb R^M} \Big[ E(\vv,\vj) - \vv\cdot \vsigma - \vj\cdot \vxi \Big].
\end{equation*}
\end{enumerate}
\end{theorem}

We call $(\vsigma,\vxi)$ \emph{ensemble $v$-representable} if there exist $\vv\in\RR^N$ and $\vj\in\RR^M$ that lead to a ground-state ensemble $\vGamma\in\DC(Q_0)$ (i.e., $\Tr \vH(\vv,\vj)\vGamma=E(\vv,\vj)$) that has $\vsigma_\vGamma=\vsigma$ 
and $\vxi_\vGamma=\vxi$.
According to \eqref{flconvprop:item:subdiff} above, every $(\vsigma,\vxi)\in(-1,1)^N\times\RR^M$ is then \emph{ensemble} $v$-representable. 
The Hohenberg--Kohn theorem implies that $\underline{\partial} F_\LI(\vsigma,\vxi)$ is a singleton, so we obtain that 
$F_\LI$ is differentiable in $\mathcal{R}_N\times\RR^M$.

The analogue of \cref{fllprop}~\eqref{fllprop:item:displacement-rule}-\eqref{fllprop:item:virial} also holds for $F_\LI$. 

\begin{theorem}[Properties of $F_\LI$]\label{flprop}
For any $(\vsigma,\vxi)\in[-1,1]^N\times\RR^M$, we have the following properties.
\begin{enumerate}[(i)]
\item\label{flprop:item:displacement} (Displacement rule) For any $\vzeta\in\RR^M$ the formula
$$
F_\LI(\vsigma,\vxi+\vzeta)=F_\LI(\vsigma,\vxi) + 2\vzeta\cdot\vxi + \vzeta \cdot \vLambda\vsigma  +  |\vzeta|^2 
$$
holds. In particular,
$$
F_\LI(\vsigma,\vxi)=F_\LI(\vsigma,\bm{0}) + \vxi \cdot \vLambda\vsigma  +  |\vxi|^2.
$$
\item\label{flprop:item:real} There is a real-valued optimizer $\vGamma$ of $F_\LI(\vsigma,\vxi)$, where real-valuedness is to be understood in the sense of kernels:
$\vGamma(\vx,\vx')=\ol{\vGamma(\vx,\vx')}$.
\item\label{flprop:item:virial} (Virial relation) For any optimizer $\vGamma$ of $F_\LI(\vsigma,\bm{0})$ the relation
$$
\Tr ((-\lapl - |\vx|^2)\vGamma)=\frac{1}{2}\Tr\left(\vx\cdot \vLambda\vsigma_z\vGamma\right)
$$
holds true.
\item\label{flprop:item:no-momentum} (Zero momentum) For any optimizer $\vGamma$ of $F_\LI(\vsigma,\vxi)$ there holds $\langle -\i\grad \rangle_{\vGamma}=0$.
\end{enumerate}
\end{theorem}

From \eqref{flprop:item:displacement}, we immediately obtain that for fixed $\vsigma\in[-1,1]^N$ the function $\vxi\mapsto F_\LI(\vsigma,\vxi)$ is a quadratic polynomial. The trivial dependence of $F_\LI(\vsigma,\vxi)$ on $\vxi$ implies a direct and simple relation between the external potential $\vj$ and the density pair $(\vsigma,\vxi)$ in the form of a ``force-balance equation''.

\begin{proposition}[Force balance]\label{prop:force-balance}
Let $\vv\in\RR^N$, $\vj\in\RR^M$, and $\vGamma\in\DC(Q_0)$ such that $\Tr \vH(\vv,\vj)\vGamma=E(\vv,\vj)$ (ground-state ensemble). Then it holds
$$\vj + \vLambda\vsigma_\vGamma + 2\vxi_\vGamma = 0.$$
\end{proposition}

Considering the adiabatic connection, the same Newton-Leibniz formula \eqref{fll-newtonleibniz} as for $F_\LL$ can be stated for $F_\LI$. Yet, the subsequent steps depend on the alteration of the coupling term, \cref{fllprop}~\eqref{fllprop:item:coupling}, which has not been proven for $F_\LI$.

\subsection{The case \texorpdfstring{$N=M=1$}{N=M=1} and unique \texorpdfstring{$v$}{v}-representability}\label{specsec} 

Consider the special case of a single two-level system and one field mode (quantum harmonic oscillator), then 
the model reduces to
\[
H_0 = (-\partial_x^2 + x^2)\iden_{\CC^2} +  \lambda x\sigma_z - t \sigma_x\qquad(t\neq 0)
\]
and is called the \emph{quantum Rabi model}. If we consider as before the Hamiltonian
\[
H(v,j) = H_0 + v\sigma_z + jx
\]
with external potentials $v\in\RR$ and $j\in\RR$, the term $jx$ can be readily absorbed by a shift $x\mapsto x-j/2$ and a constant shift in energy. With the additional $v\sigma_z$ the Hamiltonian amounts to a generalized form of the quantum Rabi model, also called `asymmetric', `driven', or `biased'~\cite{batchelor2015-AQRM,semple2017-AQRM}. An important finding for the discussion here is that the ground state of this model is always strictly positive and non-degenerate~\cite{Hirokawa2014,Nguyen2024}. In \cref{fllprop_N=M=1} below we summarize several properties 
for the Levy--Lieb functional $F_\LL$ that have been stated before or hold additionally to the ones from \cref{fllprop} for this reduced case. 
Moreover, $v$-representability is proven in \cref{uniquevrep} and the mapping from external potentials $(v,j)\in\RR^2$ to the regular set $(\sigma,\xi)\in(-1,1)\times\RR$ is even a bijection, thus implying a Hohenberg--Kohn result. Additionally, this allows to conclude in \cref{prop:F-diff-N=M=1} that the Levy--Lieb and Lieb functionals actually coincide and that they are differentiable on $(-1,1)\times\RR$.

\begin{theorem}\label{fllprop_N=M=1}
For any $(\sigma,\xi)\in [-1,1]\times\RR$, the following properties hold true for $F_\LL$.
\begin{enumerate}[(i)]
\item\label{item:sign-flip} $F_\LL(-\sigma,-\xi)=F_\LL(\sigma,\xi)$
\item\label{item:LL-prop-non-neg} There is a real-valued and non-negative optimizer of $F_\LL(\sigma,\xi)$.
\item\label{item:N=1,M=1,VR} For an optimizer $\psi$ of $F_\LL(\sigma,\xi)$ the following ``virial relation'' holds true
$$
\int \left( |\psi'|^2 - x^2|\psi|^2 \right)\dd x = \lambda \int x |\psi^+|^2 \, \dd x - \xi.
$$
\item\label{item:N=1,M=1,1stmoment} For an optimizer $\psi$ of $F_\LL(\sigma,\xi)$ there holds
$$
\int x|\psi^+|^2\,\dd x = -\int x|\psi^-|^2\,\dd x + \xi = - t \int (\psi^+)' \psi^- - \frac{\lambda(1-\sigma^2)}{4} + \frac{\xi(1+\sigma)}{2}
$$
and 
$$
\frac{1+\sigma^2}{4t}\ge \int (\psi^+)''\psi^-.
$$
\item\label{item:decoupling-FLL-explicit} For zero coupling, $\lambda=0$, there holds
$$
F_\LL(\sigma,\xi) = 1+\xi^2-t\sqrt{1-\sigma^2}.
$$
\end{enumerate}
\end{theorem}

\begin{figure}[ht]
    \includegraphics[scale=.75]{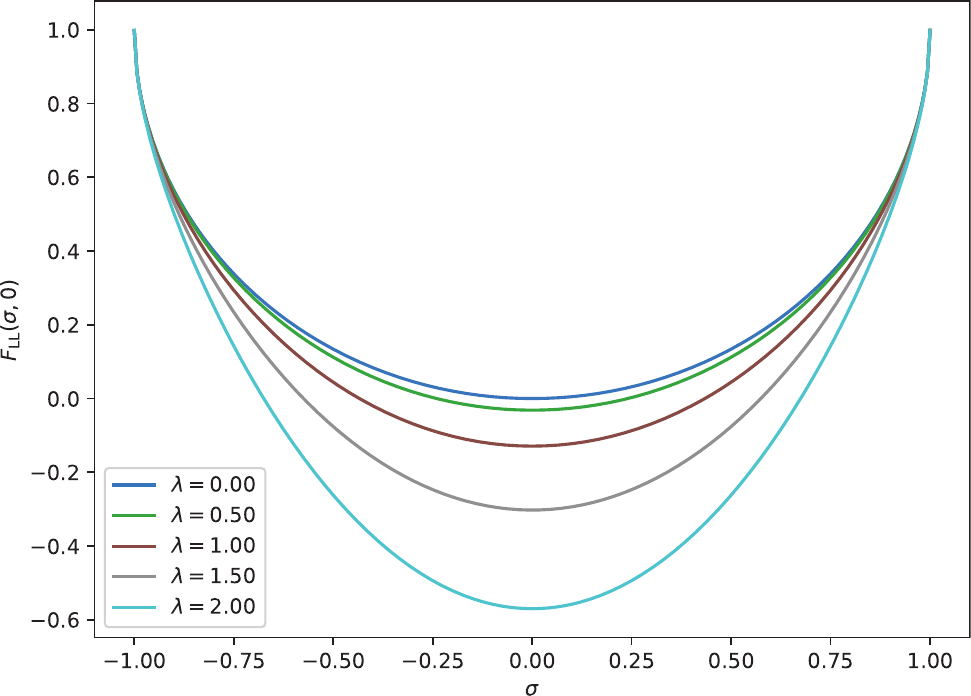}
    \centering
    \caption{The universal density functional $F_\LL(\sigma,0) = F_\LI(\sigma,0)$ in the reduced setting ($N=M=1$) with $t=1$ and for different values of $\lambda$.}
    \label{fig:QRabi-F}
\end{figure}

We continue by discussing the Euler--Lagrange equation of the constrained optimization problem in analogy to \cref{elthm}.

\begin{theorem}[Optimality]\label{thrm:EulerLagrange}
Let $\psi\in\MC_{\sigma,\xi}$ be an optimizer of $F_\LL(\sigma,\xi)$. If $\sigma\neq\pm 1$, there exist unique Lagrange multipliers $E,v,j\in\RR$ such that $\psi$ satisfies the strong Schr\"odinger equation
\begin{equation}\label{eleqs}
H(v,j)\psi=E\psi
\end{equation}
and the second-order condition
\begin{equation*}
\dua{\chi}{H(v,j)\chi}\ge E\|\chi\|^2
\end{equation*}
for all $\chi\in\TC_\psi(\MC_{\sigma,\xi})$. Moreover, $\psi$
has internal energy $\dua{\psi}{H_0\psi}=E - v\sigma - j\xi$.

If $\sigma=+ 1$, then there exists an $n\in\mathbb{N}_0$, such that $\wt{\psi}^+(x)=\psi^+(x+\xi)$ ($\psi^-\equiv 0$)
satisfies the strong Schr\"odinger equation for the harmonic oscillator instead,
\begin{equation}\label{eleqsharm}
-(\wt{\psi}^+)'' + x^2 \wt{\psi}^+ = (2n+1) \wt{\psi}^+,
\end{equation}
and one has $\dua{\psi}{H_0\psi}=2n+1 + \lambda\xi + \xi^2$. If $\sigma=-1$, 
the same result with $\psi^+$ and $\psi^-$ interchanged follows.
\end{theorem}

We have omitted the discussion of the optimality conditions for irregular $\sigma$'s in the general case, but now in the $N=M=1$ case, we see that we get a ``degenerate'' equation \cref{eleqsharm}, which may be viewed as \cref{eleqs} in the decoupling limit $t\to 0$.

Due to the aforementioned spectral properties of the quantum Rabi model, we can say much more about the optimizers than in the general case.

\begin{theorem}[unique pure-state $v$-representability]
\label{uniquevrep}
The following properties hold true.
\begin{enumerate}[(i)]
\item\label{uniquevrep:item:inner} (Regular case) If $\sigma\neq\pm 1$ then for every $(\sigma,\xi)$ there exists a unique $(v,j)$ and $\psi\in\MC_{\sigma,\xi}$ strictly positive that is the (unique) ground state of
$H(v,j)\psi=E(v,j)\psi$. Moreover, this $\psi\in \MC_{\sigma,\xi}$ is the (unique) optimizer of $F_\LL(\sigma,\xi)$.
In other words, every $(\sigma,\xi)$ pair is uniquely pure-state $v$-representable for $\sigma\neq\pm 1$.
\item\label{uniquevrep:item:endpoints}(Irregular case) If, however, $\sigma=\pm 1$ then $(\sigma,\xi)$ is not $v$-representable.
\end{enumerate}
\end{theorem}

From this, we immediately obtain that the pure-state and the mixed-state constrained-search functionals coincide.

\begin{proposition}\label{prop:F-diff-N=M=1}
There holds $F_\LI=F_\LL$ on $[-1,1]\times\RR$ and the functional is differentiable on the regular set $(-1,1)\times\RR$.
\end{proposition}

That $F_\LL$ is continuous and even differentiable is in stark contrast with standard DFT, where the corresponding functionals are everywhere discontinuous~\cite{Lammert2007}.
However, at $|\sigma|= 1$ the non-$v$-representability according to \cref{uniquevrep}~\eqref{uniquevrep:item:endpoints} implies an empty subdifferential which manifests itself with divergent derivatives of $\sigma\mapsto F_\LL(\sigma,\xi)$ as $|\sigma|\to 1$. We expect an analogous behavior also in the general case of the multi-mode Dicke model.
\Cref{fig:QRabi-F} shows the universal functional $F_\LL(\sigma,0)$ obtained from a numerical calculation for different values of the coupling constant, here denoted $\lambda$.

\section{Proofs} \label{sec:proofs}
The rest of the paper is devoted to proofs.

\subsection{Proofs of \cref{sec:Ham}}

\begin{proof}[Proof of \cref{thm:first-virial}]
Let $\vpsi$ be a ground state of $\vH(\vv,\vj)$, and let 
$$
\psi_{\vmu}^\valpha(\vx)= \mu_\valpha^{M/2} \psi^\valpha(\mu_\valpha \vx)
$$
for some $\mu_\valpha>0$. Then $\|\vpsi_{\vmu}\|=1$ and $\vpsi_{(1,\ldots,1)}=\vpsi$.
The different terms of $\dua{\vpsi_{\vmu}}{\vH(\vv,\vj)\vpsi_{\vmu}}$ read
\begin{align*}
\|\grad\vpsi_{\vmu}\|^2 &=\sum_\valpha \mu_\valpha^2 \|\grad \psi^\valpha\|^2, \\
\|\vx\vpsi_{\vmu}\|^2 &=\sum_\valpha \mu_\valpha^{-2} \|\vx\psi^\valpha\|^2, \\
\dua{\vpsi_{\vmu}}{\vx\cdot\vLambda\vsigma_z\vpsi_{\vmu}}&=\sum_{m=1}^M \sum_{n=1}^N \Lambda_{mn} \sum_{\valpha} \mu_\valpha^{-1} (\sigma_z^n)_{\valpha\valpha} \dua{\psi^\valpha}{x_m\psi^\valpha},\\
\dua{\vpsi_{\vmu}}{\vt\cdot\vsigma_x\vpsi_{\vmu}}&=\sum_{n=1}^N t_n \sum_{\valpha,\vbeta} (\sigma_x^n)_{\valpha\vbeta}
\mu_\valpha^{M/2} \mu_\vbeta^{M/2} \int_{\RR^M} \ol{\psi^\valpha(\mu_\valpha \vx)}\psi^\vbeta(\mu_\vbeta \vx)\,\dd\vx, \\
\dua{\vpsi_{\vmu}}{\vv\cdot\vsigma_z\vpsi_{\vmu}}&=\sum_{n=1}^N v_n \sum_\valpha (\sigma_z^n)_{\valpha\valpha}\|\psi^\valpha\|^2, \\
\dua{\vpsi_{\vmu}}{\vj\cdot\vx\vpsi_{\vmu}}&=\sum_{m=1}^N j_m  \sum_{\valpha}  \mu_\valpha^{-1} \dua{\psi^\valpha}{x_m\psi^\valpha}.
\end{align*}
The optimality condition
$$
0=\frac{\partial}{\partial \mu_\vbeta}\Big|_{\vmu=(1,\ldots,1)} \dua{\vpsi_{\vmu}}{\vH(\vv,\vj)\vpsi_{\vmu}}
$$
yields
\begin{align*}
0&=2\|\grad\psi^\vbeta\|^2 - 2 \|\vx\psi^\vbeta\|^2  - \sum_{n=1}^N \sum_{m=1}^M \Lambda_{mn} (\sigma_z^n)_{\vbeta\vbeta} \int_{\RR^M} x_m|\psi^\vbeta|^2 \,\dd\vx \\
&-M \sum_{n=1}^N\sum_\valpha t_n (\sigma_x^n)_{\valpha\vbeta} \Re \dua{\psi^\valpha}{\psi^\vbeta}\\
&- \sum_{n=1}^N\sum_\valpha t_n (\sigma_x^n)_{\valpha\vbeta} 2\Re \int_{\RR^M} \ol{\psi^\valpha(\vx)} \vx\cdot \grad \psi^\vbeta(\vx) \,\dd\vx\\
&-\vj\cdot\dua{\psi^\beta}{\vx\psi^\beta}.
\end{align*}
Summation over $\vbeta$ and rearranging yields
\begin{align*}
&2\|\grad\vpsi\|^2 - 2\|\vx\vpsi\|^2 - \dua{\vpsi}{\vx\cdot\vLambda\vsigma_z\vpsi}\\
&= M \dua{\vpsi}{\vt\cdot\vsigma_x\vpsi} + 2\Re \dua{\vpsi}{ (\vt\cdot\vsigma_x)(\vx\cdot\grad)\vpsi}
+\vj\cdot\dua{\vpsi}{\vx \vpsi}.
\end{align*}
Next, we use the family of wavefunctions $\vM\mapsto \vpsi_\vM(\vx)=(\det\vM)^{1/2}\vpsi(\vM\vx)$, where $\vM\in\RR^{M\times M}$ is
a diagonal matrix with positive entries $\vM=\diag(\mu_1,\ldots,\mu_M)$.
We have $\|\vpsi_\vM\|=1$. The optimality condition $\frac{\partial}{\partial\mu_m} \dua{\vpsi_{\vM}}{\vH(\vv,\vj)\vpsi_{\vM}}|_{\mu_1=\ldots=\mu_M=1}=0$
gives
$$
\|\grad\vpsi\|^2 =  \|\vx\vpsi\|^2 + \frac{1}{2}\dua{\vpsi}{\vx\cdot\vLambda\vsigma_z\vpsi} + \frac{1}{2}\vj\cdot \dua{\vpsi}{\vx \vpsi}.
$$
\end{proof}

\subsection{Proofs of \cref{sec:HK}}
\begin{proof}[Proof of \cref{prop:RN}]
We say that a subset $S$ of vertices of the $N$-cube $[-1,1]^N$ is \emph{irregular} if $\Aff(S)\neq \RR^N$,
and we say that $S$ is \emph{maximally irregular} if $\dim \Aff(S)=N-1$.
First, any irregular set is contained in some maximally irregular set. Next,
$$
\RC_N=[-1,1]^N\setminus \bigcup_{ \substack{ S\;\text{maximally} \\ \text{irregular} } } \Aff(S),
$$
so that $\RC_N$ is obtained by successively cutting the convex polytope $[-1,1]^N$ with hyperplanes, so $\RC_N$ splits into open convex polytopes.
\end{proof}


To prove \cref{hkthm,uniquevrep}, we need the following version of the unique continuation property (UCP)~\cite[Theorem 2.3]{garrigue2020unique}. Since the condition on the potential is only \emph{locally} $L^2$, it is fulfilled for the case of the harmonic oscillator potential and the coupling term.

\begin{theorem}[Strong UCP for systems of equations]\label{ucpthm}
Let $\vW\in L^2_{\rm loc}(\RR^M,\CC^{2^N\times 2^N})$ be such that for every $R>0$ there exists $c_R\ge 0$ such that
$$
\iden_{B_R} |W_{\valpha\vbeta}|^2 \le \epsilon_{N,M} (-\lapl)^{\frac{3}{2}-\delta} + c_R,
$$
for all $\valpha,\vbeta\in I$ in the sense of quadratic forms, where $\epsilon_{N,M}>0$ are small constants depending on $N$ and $M$ only. Let $\vpsi\in H^2_{\rm loc}(\RR^M,\CC^{2^N})$ be a weak solution to
$$
(-\lapl \iden_{\CC^{2^N}} + \vW)\vpsi=0.
$$
If $\vpsi$ vanishes on a set of positive measure, then $\vpsi\equiv 0$ a.e.
\end{theorem}

\begin{proof}[Proof of \cref{hkthm}]
Let $E^{(i)}=E(\vv^{(i)},\vj^{(i)})$. Using the variational principle
\begin{align*}
E^{(1)}&=\dua{\vpsi^{(1)}}{\vH(\vv^{(1)},\vj^{(1)})\vpsi^{(1)}}\le \dua{\vpsi^{(2)}}{\vH(\vv^{(1)},\vj^{(1)})\vpsi^{(2)}}\\
&=\dua{\vpsi^{(2)}}{\vH(\vv^{(2)},\vj^{(2)})\vpsi^{(2)}} + (\vv^{(1)}-\vv^{(2)})\cdot \vsigma + (\vj^{(1)}-\vj^{(2)})\cdot \vxi\\
&=E^{(2)} + (\vv^{(1)}-\vv^{(2)})\cdot \vsigma + (\vj^{(1)}-\vj^{(2)})\cdot \vxi .
\end{align*}
Moreover,
\begin{align*}
E^{(2)}&=\dua{\vpsi^{(2)}}{\vH(\vv^{(2)},\vj^{(2)})\vpsi^{(2)}}\le \dua{\vpsi^{(1)}}{\vH(\vv^{(2)},\vj^{(2)})\vpsi^{(1)}}\\
&=\dua{\vpsi^{(1)}}{\vH(\vv^{(1)},\vj^{(1)})\vpsi^{(1)}} + (\vv^{(2)}-\vv^{(1)})\cdot\vsigma + (\vj^{(2)}-\vj^{(1)})\cdot\vxi\\
&=E^{(1)} + (\vv^{(2)}-\vv^{(1)})\cdot\vsigma + (\vj^{(2)}-\vj^{(1)})\cdot\vxi .
\end{align*}
Together, these two bounds imply that there is equality everywhere, hence 
\begin{align*}
\dua{\vpsi^{(1)}}{\vH(\vv^{(2)},\vj^{(2)})\vpsi^{(1)}}&=E^{(2)},\\
\dua{\vpsi^{(2)}}{\vH(\vv^{(1)},\vj^{(1)})\vpsi^{(2)}}&=E^{(1)},
\end{align*}
so that $\vH(\vv^{(2)},\vj^{(2)})\vpsi^{(1)}=E^{(2)}\vpsi^{(1)}$ and $\vH(\vpsi^{(1)},\vj^{(1)})\vpsi^{(2)}=E^{(1)}\vpsi^{(2)}$ again by the variational principle.
This means that $\vpsi^{(1)}$ is a ground state to $\vH(\vv^{(2)},\vj^{(2)})$ as well as $\vpsi^{(2)}$ to $\vH(\vv^{(1)},\vj^{(1)})$.
Subtracting $\vH(\vv^{(i)},\vj^{(i)})\vpsi^{(i)}=E^{(i)}\vpsi^{(i)}$ from these ($i=1,2$), we obtain via \cref{hvjdef} that
\begin{equation}\label{hkpr1}
\left((E^{(2)}-E^{(1)})\iden_{\CC^{2^N}} + (\vv^{(1)}-\vv^{(2)})\cdot \vsigma_z + (\vj^{(1)}-\vj^{(2)})\cdot \vx \iden_{\CC^{2^N}} \right)_{\valpha\valpha}(\psi^{(i)})^\valpha(\vx)\equiv 0
\end{equation}
for a.e.\ $\vx\in\RR^M$ and all $\valpha\in \{+,-\}^N$. 
Letting $\vA^{(i)}(\vx)$ denote the $|I^{(i)}|\times|I^{(i)}|$ matrix in parenthesis with $I^{(i)}$ the set of spinor indices $\valpha$ for which $(\psi^{(i)})^\valpha\not\equiv 0$,
we can write \cref{hkpr1} equivalently as 
$$
\wt{\vpsi}^{(i)}(\vx)\in\ker \vA^{(i)}(\vx)
$$
for a.e.\ $\vx\in\RR^M$, where we have set $\wt{\vpsi}^{(i)}=((\psi^{(i)})^\valpha)_{\valpha\in I^{(i)}}$. 
Using \cref{ucpthm}, the functions $\wt{\vpsi}^{(i)} : \RR^M\to\CC^{I^{(i)}}$ cannot vanish on open subsets of $\RR^M$.
Hence, we obtain
\begin{equation}\label{hkspec}
0\in \sigma(\vA^{(i)}(\vx))
\end{equation}
for a.e.\ $\vx\in\RR^M$. Since $\vA(\vx)$ is simply a diagonal matrix, its spectrum reads
$$
\sigma(\vA^{(i)}(\vx))=\left\{ E^{(2)}-E^{(1)} + \sum_{n=1}^N (v^{(1)}_n-v_n^{(2)})(\sigma_z^n)_{\valpha\valpha} + (\vj^{(1)}-\vj^{(2)})\cdot \vx \right\}_{\valpha\in I^{(i)}}.
$$
Relation \cref{hkspec} then implies that there exists a function $\vchi^{(i)} : \RR^M\to \{\pm 1\}^N$ such that $(\vA^{(i)}(\vx))_{\vchi^{(i)}(\vx),\vchi^{(i)}(\vx)}=0$ a.e.\ $\vx\in\RR^M$. 
In detail,
$$
E^{(2)}-E^{(1)} + \sum_{n=1}^N (v^{(1)}_n-v_n^{(2)})(\sigma_z^n)_{\vchi^{(i)}(\vx),\vchi^{(i)}(\vx)} = - (\vj^{(1)}-\vj^{(2)})\cdot \vx,
$$
where, as a function of $\vx$, the l.h.s.\ is discontinuous or constant while the r.h.s.\ is continuous and nonconstant for $\vj^{(1)}\neq \vj^{(2)}$.
We deduce that $\vj=\vj^{(1)}=\vj^{(2)}$ and that $\vchi^{(i)}$ is constant, so $\wt{\vpsi}^{(i)}(\vx)\in\ker\vA^{(i)}$.
Hence, \eqref{hkpr1} reduces to
$$
\sum_{n=1}^N (v^{(1)}_n-v_n^{(2)})(\sigma_z^n)_{\valpha\valpha} = E^{(1)}-E^{(2)},
$$
for every $\valpha$ such that $(\psi^{(i)})^\valpha\not\equiv 0$ for some $i=1,2$. From this, the statement for the irregular case \eqref{hkthm:item:irregular} follows right away.

Now suppose that $\vsigma$ is regular and put $\vpsi=\vpsi^{(1)}$.
The preceding relation can be compactly written as a linear system
\begin{equation}\label{khprrel}
\vP\vOmega^\top(\vv^{(1)}-\vv^{(2)})=(E^{(1)}-E^{(2)})\vP\ve,
\end{equation}
where the $N\times 2^N$ matrix $\vOmega$ is given by 
$\Omega_{n,\valpha}=(\sigma_z^n)_{\valpha\valpha}$ and $\ve=(1,\ldots,1)\in \RR^{2^N}$. Also, $\vP$ is $2^N\times 2^N$ the orthogonal projector onto the linear hull of
$$
\{\ve_\valpha : \psi^\valpha\not\equiv 0\}=\{\ve_\valpha : \|\psi^\valpha\|^2\neq 0\}.
$$
The regularity of $\vsigma\in[-1,1]^N$ implies (is in fact equivalent to) that for all $\omega_\valpha\ge 0$ such that $\sum_\valpha\omega_\valpha=1$ and
$\vOmega\vomega=\vsigma$, we have $\Span \{\vOmega\ve_\valpha : \omega_\valpha\neq 0\}=\RR^N$ and $0\in\Aff \{\vOmega\ve_\valpha : \omega_\valpha\neq 0\}$.
We choose $\vomega=(\|\psi^\valpha\|^2)_{\valpha}$ in what follows.

We may distinguish two cases regarding the solvability of \cref{khprrel}.

\begin{enumerate}[(I)]
\item If $E^{(1)}=E^{(2)}$, then \cref{khprrel} is equivalent to 
$$
(\vv^{(1)}-\vv^{(2)})\in \ker \vP\vOmega^\top=(\ran \vOmega\vP)^\perp.
$$
But $\ran \vOmega\vP=\Span \{\vOmega\ve_\valpha : \omega_\valpha\neq 0\} \supseteq \Aff \{\vOmega\ve_\valpha : \omega_\valpha\neq 0\} =\RR^N$ by hypothesis, so $\vv^{(1)}=\vv^{(2)}$.

\item Suppose that $E^{(1)}\neq E^{(2)}$, which we can simply ``scale away'' in \cref{khprrel}. By the Fredholm alternative theorem,
$\vP\vOmega^\top\vu=\vP\ve$  does not have a solution $\vu\in\RR^N$ precisely if
$$
\vOmega\vP\vy=\bm{0}, \quad \vy \cdot \vP\ve\neq 0
$$
\emph{does} have a solution $\vy\in\RR^{2^N}$. In other words, if and only if there exists $\vw\in \ran \vP$ such that $\vOmega\vw=\bm{0}$ and  $\ve\cdot\vw\neq 0$.
By scaling this is equivalent to
$$
\sum_{\valpha : \psi^\valpha\not\equiv 0} w_\valpha \vOmega\ve_\valpha=0
$$
for some $w_\valpha\in\RR$ with $\sum_{\valpha : \psi^\valpha\not\equiv 0} w_\valpha=1$. In other words, $\bm{0}\in\Aff \{\vOmega\ve_\valpha : \omega_\valpha\neq 0\}$,
which holds by hypothesis.
Since the alternative has a solution, it must be that case (I) holds.
\end{enumerate}
This finishes the proof also for the regular case \eqref{hkthm:item:regular}.
\end{proof}

\subsection{Proofs of \cref{llsec}} 
\begin{proof}[Proof of \cref{repthm}]
For the $N$-representability of a given density pair $(\vsigma,\vxi)\in[-1,1]^N\times\RR^M$ define $\vpsi(\vx)= (2\pi)^{-M/4} e^{-\frac{1}{4}|\vx-\vxi|^2} \mathbf{c}$, where $\mathbf{c}\in \RR^{2^N}$ is to be determined. It is clear that $\vpsi\in Q_0$ and $|\mathbf{c}|^2=1$ is required for $\|\vpsi\|=1$. 
The constraint $\vsigma_\vpsi=\vsigma$ reads
$\sigma_z^n \mathbf{c}\cdot \mathbf{c}=\sigma_n$ for $n=1,\ldots,N$. More explicitly,
$$
\sum_\valpha |c^\valpha|^2 =1, \quad \text{and}\quad \sum_\valpha (\sigma_z^n)_{\valpha\valpha} |c^\valpha|^2=\sigma_n.
$$
These two are equivalent to finding a nonnegative solution $\beta^\valpha=|c^\valpha|^2$ of
$$
\mathbf{A} \bm{\beta}=\begin{pmatrix}
1\\
\vsigma
\end{pmatrix},
$$
where $\mathbf{A}$ is an $N+1$ by $2^N$ matrix whose first row is $(1,\ldots,1)$ and its $(n+1)$st row is the diagonal of $\sigma_z^n$.
We claim that the convex set $\{1\}\times[-1,1]^N$ is contained within the cone generated by $\mathbf{A}$, i.e., $\{\mathbf{A}\vx : x_m\ge 0\}$.
In fact, it is easy to see using the definition of the matrices $\sigma_z^n$, that the shifted hypercube $\{1\}\times[-1,1]^N$ has $2^N$ vertices and that the standard basis vectors $\mathbf{e}_\valpha$ of $\RR^{2^N}$ are mapped by $\mathbf{A}$ to these vertices.
Since $|\mathbf{c}|=1$, the constraint $\vxi_\vpsi=\vxi$ is also verified.
\end{proof}

\begin{proof}[Proof of \cref{fllexist}]
Let $\{\vpsi_j\}\subset Q_0$ be an optimizing sequence for $F_\LL(\vsigma,\vxi)$, i.e.
$\|\vpsi_j\|=1$, $\vsigma_{\vpsi_j}=\vsigma$, $\vxi_{\vpsi_j}=\vxi$ and $\dua{\vpsi_j}{\vH_0\vpsi_j}\to F_\LL(\vsigma,\vxi)$
as $j\to\infty$.
Since $\vH_0$ is bounded from below, $\{\vpsi_j\}\subset Q_0$ is bounded in the $Q_0$-norm.
Then, by the Banach--Alaoglu theorem, there exists a subsequence $\vpsi_j$ (not distinguished in notation) and $\vpsi\in Q_0$, 
such that $\vpsi_j\wconv\vpsi$ weakly in $Q_0$.
We need to show that $\vpsi$ verifies the constraints $\|\vpsi\|=1$, $\vsigma_{\vpsi}=\vsigma$, $\vxi_{\vpsi}=\vxi$ and its energy did not increase: $\dua{\vpsi}{\vH_0\vpsi}\le F_\LL(\vsigma,\vxi)$.

Since $\vV$ is a trapping potential ($|\vV(\vx)|\to\infty$ as $|\vx|\to\infty$), $-\lapl + \vV$ has compact resolvent, so
the embedding $Q_0=Q(-\lapl + \vV)\subset L^2(\RR^M,\CC^{2^N})$ is compact.
Then $\vpsi_j\wconv\vpsi$ implies that up to a subsequence (not distinguished in notation) $\vpsi_j\to \vpsi$ strongly in $L^2(\RR^M,\CC^{2^N})$, 
and hence $\|\vpsi\|=1$ and $\vsigma_{\vpsi}=\vsigma$. 
Moreover, by the Cauchy--Schwarz inequality
\begin{align*}
|\vxi_{\vpsi}-\vxi|&\le \int_{\RR^M} |\vx| |\vpsi-\vpsi_j|^2 \,\dd\vx\le \left( \int_{\RR^M} |x|^2 |\vpsi-\vpsi_j|^2\,\dd\vx\right)^{1/2} \|\vpsi_j-\vpsi\|\to 0,
\end{align*}
because the quantity in parenthesis is uniformly bounded due to the fact that  
$\{\vpsi_j\}\subset Q_0$ is bounded.

Instead of $\vpsi\mapsto\dua{\vpsi}{\vH_0\vpsi}$ we can consider $\vpsi\mapsto\dua{\vpsi}{(\vH_0+C)\vpsi}$ in the definition of $F_\LL(\vsigma,\vxi)$, which would simply yield $F_\LL(\vsigma,\vxi)+C$. But now $\vpsi\mapsto\dua{\vpsi}{(\vH_0+C)\vpsi}$ is a closed positive quadratic form, which means it is strongly lower semicontinuous in $Q_0$. By Mazur's theorem it is also weakly l.s.c., hence
$$
F_\LL(\vsigma,\vxi)+C=\lim_{j\to\infty} \dua{\vpsi_j}{(\vH_0+C)\vpsi_j}\ge \dua{\vpsi}{(\vH_0+C)\vpsi}.
$$
But we already know that $\|\vpsi\|=1$, so $C$ cancels from the above inequality and we obtain that $\vpsi$ is 
an optimizer for $F_\LL(\vsigma,\vxi)$. 
\end{proof}

\begin{proof}[Proof of \cref{fllprop}]
First, we note that the quadratic energy functional $\EC : Q_0 \to \RR$ that enters the definition of $F_\LL$ can be written in a convenient form as
$$
\EC(\vpsi)=\dua{\vpsi}{\vH_0\vpsi}=\|\grad\vpsi\|^2 + \|\vx\vpsi\|^2 + \dua{\vpsi}{\vx\cdot \vLambda\vsigma_z\vpsi} - \dua{\vpsi}{\vt\cdot\vsigma_x\vpsi}.
$$
For part \eqref{fllprop:item:displacement-rule}, consider the shift operator $\DC_\vzeta \vpsi(\vx)=\vpsi(\vx-\vzeta)$ for any $\vzeta\in\RR^M$. Clearly, $\|(\DC_\vzeta\psi)^\valpha\|^2=\|\psi^\valpha\|^2$, $\vsigma_{\DC_\vzeta \vpsi}=\vsigma_\vpsi$, $\vxi_{\DC_\vzeta \vpsi}=\vxi_\vpsi + \vzeta$ and furthermore
$$
\EC(\DC_\vzeta \vpsi)= \EC(\vpsi) + 2\vzeta\cdot\vxi + \vzeta\cdot\vLambda\vsigma + |\vzeta|^2.
$$
Using $\DC_\vzeta \vpsi$ as a trial state for $F_\LL(\vsigma,\vxi+\vzeta)$, where $\vpsi$ is an optimizer for $F_\LL(\vsigma,\vxi)$, we obtain the stated relation with a ``$\le$''.
Conversely, choosing $\DC_{-\vzeta}\vpsi$ as a trial state for $F_\LL(\vsigma,\vxi)$, where $\vpsi$ is an optimizer for $F_\LL(\vsigma,\vxi+\vzeta)$, we get the opposite inequality.

Next, for the real-valuedness part of \eqref{fllprop:item:real}, it is enough to note that we may decouple the real and the imaginary parts of $\vpsi$ within the constraints and the energy. The only nontrivial terms are
\begin{align*}
\dua{\vpsi}{\vx\cdot\vLambda\vsigma_z\vpsi}&=\sum_{m=1}^M\sum_{n=1}^N \Lambda_{mn}\dua{\Re\vpsi}{x_m\sigma_z^n\Re\vpsi} + \sum_{m=1}^M\sum_{n=1}^N \Lambda_{mn}\dua{\Im\vpsi}{x_m\sigma_z^n\Im\vpsi}
\end{align*}
and
\begin{align*}
\dua{\vpsi}{\vt\cdot\vsigma_x\vpsi}&=\sum_{n=1}^N t_n \dua{\vpsi}{\sigma_x^n\vpsi}=
\sum_{n=1}^N t_n \dua{\Re\vpsi}{\sigma_x^n\Re\vpsi} + \sum_{n=1}^N t_n \dua{\Im\vpsi}{\sigma_x^n\Im\vpsi},
\end{align*}
where the mixed terms cancel using the fact that both $\sigma_z^n$ and $\sigma_x^n$ are real symmetric.
This allows to minimize $\EC(\vpsi) = \EC(\Re\vpsi) + \EC(\Im\vpsi)$ with just $\Re\vpsi$.

To see the virial relation \eqref{fllprop:item:virial}, we employ the usual scaling argument~\cite{LevyPerdewPhysRevA1985}. 
Consider the family of wavefunctions $\vM\mapsto \vpsi_\vM(\vx):=(\det\vM)^{1/2}\vpsi(\vM\vx)$, where $\vM\in\RR^{M\times M}$ is
a diagonal matrix with positive entries $\vM=\diag(\mu_1,\ldots,\mu_M)$ and $\vpsi$ is an optimizer for $F_\LL(\vsigma,\bm{0})$.
We have $\|\vpsi_\vM\|=1$, $\vsigma_{\vpsi_\vM}=\vsigma$ and $\vxi_{\vpsi_\vM}=\vM^{-1}\vxi=0$. Moreover,
$$
\EC(\vpsi_\vM)= \|\vM\grad\vpsi\|^2 + \|\vM^{-1}\vx\vpsi\|^2
+ \dua{\vpsi}{\vM^{-1}\vx\cdot \vLambda\vsigma_z\vpsi} -
\dua{\vpsi}{\vt\cdot\vsigma_x\vpsi}.
$$
Since $\vpsi$ is an optimizer, we necessarily have $\frac{\partial}{\partial\mu_m} \EC(\vpsi_\vM)|_{\mu_1=\ldots=\mu_M=1}=0$, i.e., with $\ve_m$ the $m$th unit vector in $\RR^M$,
$$
0=2 \|\partial_{x_m}\vpsi\|^2 - 2\|x_m\vpsi\|^2 - \dua{\vpsi}{x_m\ve_m\cdot \vLambda\vsigma_z\vpsi},
$$
from which the stated formula follows.

To prove \eqref{fllprop:item:coupling}, fix $m$ and consider the shifted family 
$$
s\mapsto \vpsi_s(\vx)=( \psi^\valpha(\vx+ sc^\valpha\ve_m) )_{\valpha},
$$
where $\vpsi$ is a real-valued optimizer (according to \eqref{fllprop:item:real}) for $F_\LL(\vsigma,\vxi)$ and $c^\valpha$ is given by
$$
c^\valpha=\sum_{n=1}^N \Lambda_{mn}(\sigma_z^n)_{\valpha\valpha} - \sum_{n=1}^N \Lambda_{mn}\sigma_n
$$
so that
$$
\sum_{\valpha} c^\valpha \|\psi^\valpha\|^2=0.
$$
We have $\|\vpsi_s\|=1$,  $\vsigma_{\vpsi_s}=\vsigma$ and 
\begin{align*}
\vxi_{\vpsi_s}&=\sum_{\valpha}\int_{\RR^M} \vx |\psi^\valpha(\vx+ sc^\valpha\ve_m)|^2 \,\dd\vx= \sum_{\valpha}\int_{\RR^M} (\vx-sc^\valpha\ve_m) |\psi^\valpha|^2 \,\dd\vx=\vxi
\end{align*}
by construction. Now
\begin{align*}
\EC(\vpsi_s)&=\dua{\vpsi_s}{\vH_0\vpsi_s}=\|\grad\vpsi\|^2 + \sum_\valpha\int_{\RR^M}|\vx-sc^\valpha\ve_m|^2|\psi^\valpha|^2\,\dd\vx \\
&+ \sum_{n'=1}^N\sum_{m'=1}^M\Lambda_{m'n'}\sum_{\valpha}\int_{\RR^M} (x_{m'}-sc^\valpha\delta_{mm'})(\sigma_z^{n'})_{\valpha\valpha}|\psi^\valpha|^2\,\dd\vx\\
&- \sum_{n'=1}^N t_{n'}\sum_{\valpha,\vbeta}(\sigma_x^{n'})_{\vbeta\valpha} \int_{\RR^M} \psi^\vbeta(\vx+ sc^\vbeta\ve_m)\psi^\valpha(\vx+ sc^\valpha\ve_m)\,\dd\vx,
\end{align*}
and the optimality condition $\frac{\dd}{\dd s} \EC(\vpsi_s)|_{s=0}=0$ implies
\begin{align*}
0=&-2\sum_\valpha\int_{\RR^M} c^\valpha x_m|\psi^\valpha|^2 \,\dd\vx - \sum_{n'=1}^N\Lambda_{mn'}\sum_{\valpha}(\sigma_z^{n'})_{\valpha\valpha} c^\valpha \|\psi^\valpha\|^2\\
&- \sum_{n'=1}^N t_{n'}\sum_{\valpha,\vbeta}(\sigma_x^{n'})_{\vbeta\valpha} c^\vbeta \int_{\RR^M} (\partial_{x_m}\psi^\vbeta(\vx)) \psi^\valpha(\vx)\,\dd\vx\\
&- \sum_{n'=1}^N t_{n'}\sum_{\valpha,\vbeta}(\sigma_x^{n'})_{\vbeta\valpha} c^\valpha \int_{\RR^M} \psi^\vbeta(\vx) (\partial_{x_m}\psi^\valpha(\vx)) \,\dd\vx.
\end{align*}
Using the definition of $c^\valpha$ and the symmetry of $\sigma_x^n$ for the last term, we find
\begin{align*}
0=&-2\sum_\valpha\sum_{n=1}^N \int_{\RR^M} x_m \Lambda_{mn} (\sigma_z^n)_{\valpha\valpha} |\psi^\valpha|^2 \,\dd\vx + 2\sum_{n=1}^N \Lambda_{mn}\sigma_n \sum_\valpha\int_{\RR^M} x_m|\psi^\valpha|^2 \,\dd\vx \\
&- \sum_{n,n'=1}^N\Lambda_{mn'}\Lambda_{mn} \sum_{\valpha} (\sigma_z^{n'})_{\valpha\valpha}  (\sigma_z^n)_{\valpha\valpha} \|\psi^\valpha\|^2 + \left(\sum_{n=1}^N \Lambda_{mn}\sigma_n\right)^2  \\
&- 2\sum_{n,n'=1}^N t_{n'}\Lambda_{mn} \sum_{\valpha,\vbeta} (\sigma_x^{n'})_{\vbeta\valpha}\left((\sigma_z^n)_{\vbeta\vbeta}- \sigma_n \right) \int_{\RR^M} (\partial_{x_m}\psi^\vbeta(\vx))\psi^\valpha(\vx)\,\dd\vx
\end{align*}
from which
\begin{align*}
0=&-2 \dua{\vpsi}{x_m\ve_m\cdot \vLambda\vsigma_z\vpsi} + 2\sum_{n=1}^N \xi_m\Lambda_{mn}\sigma_n \\
&- \sum_{n,n'=1}^N\Lambda_{mn'}\Lambda_{mn} \dua{\sigma_z^{n'}\vpsi}{\sigma_z^{n}\vpsi} + \left(\sum_{n=1}^N \Lambda_{mn}\sigma_n\right)^2 \\
&- 2\sum_{n,n'=1}^N t_{n'}\Lambda_{mn}\dua{\sigma_x^{n'}\vpsi}{(\sigma_z^n-\sigma_n)\partial_{x_m}\vpsi}.
\end{align*}
The statement then follows from summation over $m$.
\end{proof}

\begin{proof}[Proof of \cref{tangthm}]
Consider the smooth map $\vf : Q_0\to \RR^{N+M+1}$ given by 
$$
\vf(\vpsi)=(\|\vpsi\|^2,\vsigma_\psi,\vxi_\vpsi).
$$
Then $\vf^{-1}(1,\vsigma,\vxi)=\MC_{\vsigma,\vxi}$. Note that
$$
\vf'(\vpsi)\vchi=2(\dua{\vpsi}{\vchi}, \dua{\sigma_z^1\vpsi}{\vchi}, \ldots, \dua{\sigma_z^N\vpsi}{\vchi}, \dua{x_1\vpsi}{\vchi},
\ldots, \dua{x_M\vpsi}{\vchi}).
$$
We claim that the differential $(\dd\vf)_\vpsi : Q_0\simeq \TC_\vpsi(Q_0) \to \TC_{(1,\vsigma,\vxi)}(\RR^{N+M+1})\simeq \RR^{N+M+1}$
of $\vf$ at any $\vpsi\in \vf^{-1}(1,\vsigma,\vxi)$ is surjective. To see this, we show that the functions
$$
\{\vpsi\}\cup\{\sigma_z^1\vpsi,\ldots,\sigma_z^N\vpsi\}\cup \{x_1\vpsi,\ldots,x_M\vpsi\}
$$
are linearly independent. Assume for contradiction that there are coefficients $a_0$, $a_1,\ldots,a_N\in\RR$ and $b_1,\ldots,b_M\in\RR$ (not all zero) such that
\begin{equation}\label{lindeppr}
\left(a_0 + \sum_{n=1}^N a_n (\sigma_z^n)_{\valpha\valpha} + \sum_{m=1}^M b_m x_m \right)\psi^\valpha(\vx)=0
\end{equation}
for all $\vx\in\RR^M$ and $\valpha\in\{+,-\}^N$.
Differentiating \cref{lindeppr} with respect to $x_p$, $p=1,\ldots,M$, we find
\begin{equation}\label{lindeppr2}
\left(a_0 + \sum_{n=1}^N a_n (\sigma_z^n)_{\valpha\valpha} + \sum_{m=1}^M b_m x_m \right)\partial_{x_p}\psi^\valpha(\vx)=-b_p \psi^\valpha(\vx).
\end{equation}
We can assume that $b_p\neq 0$ for at least one $p$, since if $b_1=\ldots=b_M=0$, then \cref{lindeppr} implies
$$
\vP\vOmega^\top\va=a_0\vP\ve,
$$
using the same notations as in the Hohenberg--Kohn theorem's proof, from which $a_0=0$ and $\va=\bm{0}$ follows by the same argument using the regularity of $\vsigma$.

Plugging \cref{lindeppr2} into \cref{lindeppr}, we find
$$
\left(a_0 + \sum_{n=1}^N a_n (\sigma_z^n)_{\valpha\valpha} + \sum_{m=1}^M b_m x_m \right)^2 \partial_{x_p}\psi^\valpha(\vx) =0
$$
for all $\vx\in\RR^M$ and $p=1,\ldots,M$ such that $b_p\neq 0$. Taking the inner product with $\partial_{x_p}\psi^\valpha(\vx)$, we find
$$
\left(a_0 + \sum_{n=1}^N a_n (\sigma_z^n)_{\valpha\valpha} + \sum_{m=1}^M b_m x_m \right)\partial_{x_p}\psi^\valpha(\vx)=0.
$$
Using \cref{lindeppr2}, we find $b_p\psi^\valpha(\vx)=0$. But then $b_p=0$ 
which is a contradiction.

Therefore, $(1,\vsigma,\vxi)$ is a regular value of $\vf$.
We conclude using the submersion theorem 
(see e.g.\ \cite[Theorem 3.5.4]{abraham2012manifolds}) that $\MC_{\vsigma,\vxi}$ is a closed submanifold of $Q_0$, which shows \eqref{tangthm:item:submersed},
and that the tangent space of $\MC_{\vsigma,\vxi}$ is given by $\TC_\vpsi(\MC_{\vsigma,\vxi})=\ker (\dd\vf)_\vpsi$,
from which the stated formulas in \eqref{tangthm:item:tangent-space} and \eqref{tangthm:item:cotangent-space} follow easily.
\end{proof}


\begin{proof}[Proof of \cref{elthm}]
Let $\mathcal{E}(\vpsi)=\dua{\vpsi}{\vH_0\vpsi}$, then the criticality condition reads
\begin{equation}\label{elpr}
0=(\dd\mathcal{E}|_{\MC_{\vsigma,\vxi}})_\vpsi(\vchi)=\dua{\vchi}{\vH_0\vpsi}
\end{equation}
for all $\vchi\in\TC_\vpsi(\MC_{\vsigma,\vxi})$. 
Using \cref{tangthm}~\eqref{tangthm:item:cotangent-space} this implies that there exists Lagrange multipliers $E\in\RR$, $\vv\in\RR^N$ and $\vj\in\RR^M$ such that
$$
\vH_0\vpsi=E\vpsi - \vv\cdot\vsigma_z \vpsi - \vj\cdot \vx\vpsi
$$
holds true in $Q_0^*$. Recall that $Q_0$ is a dense subspace of $\HC$, and so the preceding equation holds true strongly in $\HC$ as well. Hence, $\vH_0\vpsi\in\HC$ so $\vpsi$ is in the domain of $\vH(\vv,\vj)$.

To show the second order condition \cref{seccond}, let $\vgamma : (-1,1)\to \MC_{\vsigma,\vxi}$ be a $C^2$-curve such that $\vgamma(0)=\vpsi$ and $\dot{\vgamma}(0)=\vchi$. Differentiating the constraints
$\|\vgamma(t)\|^2=1$, $\vsigma_{\vgamma(t)}=\vsigma$ and $\vxi_{\vgamma(t)}=\vxi$ twice, we obtain 
\begin{equation}\label{constrdiff}
\begin{aligned}
2\Re\dua{\ddot{\vgamma}(0)}{\vpsi}&=- 2\|\vchi\|^2,\\
2\Re\dua{\ddot{\vgamma}(0)}{\vsigma_z^n\vpsi}&=-2\dua{\vchi}{\vsigma_z^n\vchi}\quad (n=1,\ldots,N),\\
2\Re\dua{\ddot{\vgamma}(0)}{x_m\vpsi}&=-2\dua{\vchi}{x_m\vchi}\quad (m=1,\ldots,M).
\end{aligned}
\end{equation}
We have
\begin{align*}
(\dd^2\EC)_\vpsi(\vchi,\vchi)&=\frac{\dd^2}{\dd t^2}\Big|_{t=0} \dua{\vgamma(t)}{\vH_0\vgamma(t)}= 2 \dua{\vchi}{\vH_0\vchi} + 2\Re\dua{\ddot{\vgamma}(0)}{\vH_0\vpsi}.
\end{align*}
Using $\vH(\vv,\vj)=\vH_0 + \vv\cdot \vsigma_z + \vj\cdot \vx$, we get
\begin{align*}
(\dd^2\EC)_\vpsi(\vchi,\vchi)&=2 \dua{\vchi}{\vH(\vv,\vj)\vchi} + 2\Re\dua{\ddot{\vgamma}(0)}{\vH(\vv,\vj)\vpsi}\\
&- 2\sum_{n=1}^N v_n\dua{\vchi}{\vsigma_z^n\vchi} - 2\sum_{m=1}^M j_m\dua{\vchi}{x_m\vchi} \\
&- 2\Re\sum_{n=1}^N v_n\dua{\ddot{\vgamma}(0)}{\vsigma_z^n\vpsi} - 2\Re\sum_{m=1}^M j_m\dua{\ddot{\vgamma}(0)}{x_m\vpsi} \\
&=2 \dua{\vchi}{(\vH(\vv,\vj)-E)\vchi},
\end{align*}
for all $\vchi\in\TC_\vpsi(\MC_{\vsigma,\vxi})$, where we used \cref{constrdiff} to cancel several terms.
Since the Hessian needs to be nonnegative at the optimizer $\vpsi$,
we obtain the stated relation.
\end{proof}

\begin{proof}[Proof of \cref{hkopt}]
The result follows from a similar argument as given in the proof of the Hohenberg--Kohn theorem (\cref{hkthm}) that allows us to conclude that ground states that share the same $(\vsigma,\vxi)$ are also ground states for each others Hamiltonians.
\end{proof}

To prove \cref{aufbauthm}, we need a well-known lemma.

\begin{lemma}\cite[Lemma~II.2]{lions1987solutions}\label{lionslemm}
Let $H$ be a Hilbert space, $A:H\to H$ a self-adjoint operator bounded from below. If $A$ is nonnegative on a subspace of codimension $d$,
then $A$ has at most $d$ nonpositive eigenvalues.
\end{lemma}
\begin{proof}
Let $\mu_1\le\mu_2\le\ldots$ denote the min-max values of $A$. If $A$ has at most $d$ eigenvalues below the bottom of its essential spectrum (which can be $+\infty$),
then we are done. Otherwise, $\mu_1,\ldots,\mu_{d+1}$ are eigenvalues and the max-min formula implies that
$$
\mu_{d+1}\ge \inf_{\substack{u\in W\\\|u\|=1}} \dua{Au}{u}\ge 0,
$$
where $W\subset H$ is the $d$-codimensional subspace on which $A$ is nonnegative.
\end{proof}

\begin{proof}[Proof of \cref{aufbauthm}]
The condition \cref{seccond} implies that $\vH(\vv,\vj)-E$ is nonnegative on the $(N+M+1)$-codimensional subspace $\TC_\vpsi(\MC_{\vsigma,\vxi})$.
Moreover, we have that $\dua{\vpsi}{(\vH(\vv,\vj)-E)\vpsi}=0$
by \cref{elschr} and $\vpsi\in \TC_\vpsi(\MC_{\vsigma,\vxi})^\perp$, so 
actually $\vH(\vv,\vj)-E$ is nonnegative on an $(N+M)$-codimensional subspace.
\Cref{lionslemm} then implies that $\vH(\vv,\vj)-E$ has at most $N+M$ nonpositive eigenvalues.
\end{proof}

\subsection{Proofs of \cref{sec:AC}}

\begin{proof}[Proof of \cref{lemma:FLL0}] 
By virtue of $\vLambda = 0$, the $N$ two-level systems decouple. Thus, the two-level part of the wavefunction can be written as the $N$-fold tensor product of independent two-level wavefunctions, i.e.,
\begin{equation*}
    \vpsi(\vx) = \bigotimes_{n=1}^N \left(\psi^+_n,\, \psi^-_n \right)^\top \in L^2(\RR^M, \CC^{2^N}).
\end{equation*} 
With this simplification, the constraints for an optimizer $\vpsi$
can be combined into the equivalent constraints
\begin{align*}
\norm{\psi_i^\pm}^2 = \frac{1 \pm \sigma_i}{2} \left[\prod_{\substack{n=1\\n\neq i}}^N \left( \norm{\psi_n^+}^2 + \norm{\psi_n^-}^2\right) \right]^{-1}.
\end{align*}
By use of these constraints and the Cauchy--Schwarz inequality 
it follows that $\langle\vpsi, \sigma_x^i \vpsi \rangle \leq  \sqrt{1-\sigma_i^2}$.
Furthermore, $\norm{\grad \vpsi}^2 + \norm{\vx\vpsi}^2\geq M $ by the ground-state energy of the independent harmonic oscillators,  
implying the lower bound 
$\langle\vpsi, \vH_0 \vpsi\rangle \geq M - \sum_{n=1}^N t_n \sqrt{1-\sigma_n^2}.$
Consequently, $F_\mathrm{LL}^0(\vsigma,\vxi) \geq M - \sum_{n=1}^N t_n \sqrt{1-\sigma_n^2}$, and  by the displacement rule, \cref{fllprop}~\eqref{fllprop:item:displacement-rule}, that
\begin{equation*}
F_\mathrm{LL}^0(\vsigma,\vxi)  \geq M - \sum_{n=1}^N t_n \sqrt{1-\sigma_n^2} + |\vxi|^2.
\end{equation*}
Similarly to the proof of \cref{repthm}, suppose the trial state $\psi(\vx) = \pi^{-M/4} e^{-|\vx - \vxi|^2/2} \vc$ with $\vc\in\CC^{2^N}$ left unspecified. 
Then by the same calculation, it follows that $|\vc|^2=1$, $\vxi_\vpsi = \vxi$ and $ \langle \vc, \vsigma_z \vc \rangle =\vsigma$.
Furthermore, we have that $\norm{\grad\vpsi}^2 =\frac{M}{2}$ and $\norm{\vx \vpsi}^2 = \frac{M}{2} + |\vxi|^2$.
In the case $\vLambda=\bm{0}$, as discussed above, it is sufficient to consider $\vc \in \CC^{2N}$. Suppose the particular choice
of $\vc$ in \eqref{eq:FLL0TrialState}, then it immediately follows that 
the constraints are satisfied, and that  $ \langle \vc, \sigma_x^i \vc \rangle = \sqrt{1-\sigma_i^2}$.
Then
\begin{equation*}
\langle \vpsi, \vH_0 \vpsi\rangle = M + |\vxi|^2 - \sum_{n=1}^N t_n \sqrt{1-\sigma_n^2},
\end{equation*}
which in fact equals the lower bound for $F^0_\mathrm{LL}(\vsigma,\vxi)$,
i.e., $\vpsi$ is a minimizer.
\end{proof}

\begin{proof}[Proof of \cref{lemma-FLL-subdiff}]
For any optimizers $\vpsi_{s\vLambda}$ and $\vpsi_{s'\vLambda}$ of $F_\LL^{s\vLambda}(\vsigma,\vxi)$ and $F_\LL^{s'\vLambda}(\vsigma,\vxi)$ respectively, we have
\begin{align*}
F_\LL^{s'\vLambda}(\vsigma,\vxi)&=\dua{\vpsi_{s'\vLambda}}{\vH_0^{s'\vLambda}\vpsi_{s'\vLambda}}\le \dua{\vpsi_{s\vLambda}}{\vH_0^{s'\vLambda}\vpsi_{s\vLambda}}\\
&=\dua{\vpsi_{s\vLambda}}{\vH_0^{s\vLambda}\vpsi_{s\vLambda}} + \dua{\vpsi_{s\vLambda}}{\vx\cdot(s'-s)\vLambda\vsigma_z\vpsi_{s\vLambda}}\\
&=F_\LL^{s\vLambda}(\vsigma,\vxi) + (s'-s)\dua{\vpsi_{s\vLambda}}{\vx\cdot\vLambda\vsigma_z\vpsi_{s\vLambda}}.
\end{align*}
This implies that $\dua{\vpsi_{s\vLambda}}{\vx\cdot\vLambda\vsigma_z\vpsi_{s\vLambda}} \in \ol{\partial}_s F_\LL^{s\vLambda}(\vsigma,\vxi)(s)$.
\end{proof}

\subsection{Proofs of \cref{flsec}}

\begin{proof}[Proof of \cref{flexsist}]
The proof closely follows the one of Theorem~4.4 in \cite{Lieb1983}.
By shifting $\vH_0$ to $\vH_0+C$ as in the proof of \cref{fllexist}, we can achieve $\vH_0+C>0$. Set $\vh=\sqrt{\vH_0+C}$,
then $\vh^{-1}$ is compact, because the resolvent $(\vH_0+C)^{-1}$ is. Let $\{\vGamma_j\}\subset \DC(Q_0)$ be a minimizing sequence such that $\vsigma_{\vGamma_j}=\vsigma$, $\vxi_{\vGamma_j}=\vxi$. Then for any $\varepsilon>0$ there is a $J\in\mathbb{N}$ such that
$$
\Tr \vh\vGamma_j\vh=\Tr \vh^2\vGamma_j\le F_\LI(\vsigma,\vxi) -C+\varepsilon
$$
for all $j\geq J$. We will from here on switch to the subsequence indexed by $j\geq J$. The estimate above implies that the sequence $\vT_j=\vh\vGamma_j\vh$ is bounded in trace norm and is thus in $\mathfrak{S}_1(\HC)$.
Since $\mathfrak{S}_\infty(\HC)^*=\mathfrak{S}_1(\HC)$, with the dual pairing $\Tr \vT\vK$, where $\vT\in\mathfrak{S}_1(\HC)$ and $\vK\in\mathfrak{S}_\infty(\HC)$,
the Banach--Alaoglu theorem implies that there exists $\vT\in\mathfrak{S}_1(\HC)$, such that up to a subsequence,
\begin{equation}\label{flprconv}
\Tr \vT_j\vK \to \Tr \vT\vK\quad \text{for all}\quad \vK\in \mathfrak{S}_\infty(\HC).
\end{equation}
Since $\vT\ge 0$, we have $\liminf_{j\to\infty}\Tr (\vH_0+C)\vGamma_j=\liminf_{j\to\infty} \Tr \vT_j\ge \Tr \vT=\Tr (\vH_0+C)\vGamma$,
where we have set $\vGamma=\vh^{-1}\vT\vh^{-1}$.

We need to prove now that the self-adjoint, positive trace-class operator $\vGamma$ with finite $\vH_0$-energy
has the right constraints: $\Tr\vGamma=1$, $\vsigma_{\vGamma}=\vsigma$ and $\vxi_{\vGamma}=\vxi$.
First,  taking $\vK=\vh^{-2}$ in \cref{flprconv}, we find $1=\Tr\vGamma_j\to \Tr \vT\vh^{-2}=\Tr\vGamma$.
Next, we put $\vK=\vh^{-1} \sigma_z^n\vh^{-1}$, from which follows
$$
\sigma_n=\Tr \sigma_z^n\vGamma_j=\Tr \vT_j \vh^{-1} \sigma_z^n\vh^{-1} \to \Tr \vT \vh^{-1} \sigma_z^n\vh^{-1}=\Tr \sigma_z^n\vGamma,
$$
i.e.\ $\vsigma_{\vGamma}=\vsigma$.
Lastly, we choose $\vK=\vh^{-1} x_m \iden_{\CC^{2^N}} \vh^{-1}$, which is compact since $x_m$ is bounded in $Q_0$ and the embedding $Q_0\subset\HC$ is compact. Then
$$
\xi_m=\Tr x_m\vGamma_j=\Tr \vT_j \vh^{-1} x_m \iden_{\CC^{2^N}} \vh^{-1}\to \Tr \vT \vh^{-1} x_m \iden_{\CC^{2^N}} \vh^{-1}=\Tr x_m\vGamma,
$$
which finishes the proof.
\end{proof}

\begin{proof}[Proof of \cref{flconvprop}]
Part \eqref{flconvprop:item:lsc} can be proven like Theorem~4.4 and Corollary~4.5 in \cite{Lieb1983} and a similar proof is given here for \cref{flexsist}. Thus this will not be repeated here.

For the proof of \eqref{flconvprop:item:convex} we start from the definition of $F_\LI$,
$$
F_\LI(\vsigma,\vxi)=\inf_{\substack{\vGamma\in\DC(Q_0)\\ \vsigma_\vGamma=\vsigma\\ \vxi_\vGamma=\vxi}} \Tr(\vH_0\vGamma).
$$
Here, $\vGamma$ can always be written as an (infinite) convex combination of pure-state projectors as in \eqref{eq:Gamma-set}. With this and $\vsigma_j = \vsigma_{\vpsi_j}$, $\vxi_j=\vxi_{\vpsi_j}$ we have
$$
\Tr(\vH_0\vGamma) = \sum_j c_j \dua{\vpsi_j}{\vH_0\vpsi_j}, \quad \vsigma=\sum_j c_j \vsigma_{j}, \quad \vxi=\sum_j c_j \vxi_{j},\quad \sum_j c_j=1.
$$
So instead of over all $\vGamma$, the infimum can be taken first over all $c_j\in\RR_+$, $\vsigma_{j}\in[-1,1]^N$, and $\vxi_j\in\RR^M$ under the constraints above before then taking another infimum over all possible pure states that have $\vsigma_{\vpsi_j}=\vsigma_j$, $\vxi_{\vpsi_j}=\vxi_j$, i.e., they are from the constraint manifold $\MC_{\vsigma_j,\vxi_j}$, and further fulfil the constraints from \eqref{eq:Gamma-set}. Without denoting all the constraints this leads to
$$
F_\LI(\vsigma,\vxi) = \inf_{\vsigma_j,\vxi_j,c_j} \sum_j c_j  \inf_{\vpsi_j\in\MC_{\vsigma_j,\vxi_j}} \dua{\vpsi_j}{\vH_0\vpsi_j} = \inf_{\vsigma_j,\vxi_j,c_j} \sum_j c_j F_\LL(\vsigma_j,\vxi_j).
$$
Note that we were able to move the convex sum outside of the inner infimum and thus arrive exactly at the definition for $F_\LL$. But this expression is nothing else but the convex envelope. Since we now have that $F_\LI$ is convex and since the existence of an optimizer (\cref{flexsist}) also gives that $F_\LI$ is locally bounded on all of $(-1,1)^N\times\RR^M$, it is also locally Lipschitz on this set.

To see \eqref{flconvprop:item:subdiff}, we use a well-known characterization: $(-\vv,-\vj)\in\underline{\partial} F_\LI(\vsigma,\vxi)$ if and only if
the convex functional $(\vsigma',\vxi')\mapsto F_\LI(\vsigma',\vxi') + \vv\cdot\vsigma' + \vj\cdot\vxi'$ has global minimum at $(\vsigma,\vxi)$. 
By \eqref{eq:FLprop:E-inf}, this in turn is equivalent to 
\begin{align*}
E(\vv,\vj)&=F_\LI(\vsigma,\vxi) + \vv\cdot\vsigma + \vj\cdot\vxi\\
&=\Tr \vH_0\vGamma + \vv\cdot\vsigma + \vj\cdot\vxi=\Tr \vH(\vv,\vj)\vGamma
\end{align*}
for some $\vGamma\in\DC(Q_0)$ with $\vsigma_\vGamma=\vsigma$ and $\vxi_{\vGamma}=\vxi$. 

The last item, \eqref{flconvprop:item:Legendre},  is a direct consequence of how $E$ is given as the transformation of $F_\LI$ in \eqref{eq:FLprop:E-inf} and the fact that since $F_\LI$ is both convex and lower semicontinuous it must agree with the back-transformation of $E$. 
\end{proof}

\begin{proof}[Proof of \cref{flprop}]
We first show \eqref{flprop:item:displacement}. A proof like for \cref{fllprop}~\eqref{fllprop:item:displacement-rule} is possible here, but we will show another technique that employs the ladder operators $a_m^\dag,a_m$. Define
$a_m^\dag=(x_m-\partial_{x_m})/\sqrt{2}$ and $a_m=(x_m+\partial_{x_m})/\sqrt{2}$. Note that we have $[a_m,a_{m'}]=0$ and $[a_m,a_{m'}^\dag]=[a_m^\dag,a_{m'}]=0$ ($m\neq m'$), and $[a_m,a_m^\dag]=1$. For any $\vzeta\in\CC^M$ define the multimode displacement operator
$$
D(\vzeta)=\prod_{m=1}^M \exp\left( \frac{\zeta_m a_m^\dag - \ol{\zeta_m}a_m}{\sqrt{2}} \right).
$$
Clearly, $D(\vzeta)^{-1}=D(-\vzeta)=D(\vzeta)^\dag$, i.e.\ $D(\vzeta)$ is unitary.
Moreover,
\begin{equation}\label{displadder}
\begin{aligned}
D(\vzeta)^\dag a_m D(\vzeta)&=a_m + \zeta_m / \sqrt{2},\\
D(\vzeta)^\dag a_m^\dag D(\vzeta)&=a_m^\dag + \ol{\zeta}_m / \sqrt{2}.
\end{aligned}
\end{equation}
Notice that $\vxi_\vGamma$ can be alternatively written as 
$$
(\vxi_\vGamma)_m=\Tr\left( \frac{a_m + a_m^\dag}{\sqrt{2}}\sps{\Gamma} \right),
$$
so that 
$$
\vxi_{D(\vzeta) \vGamma D(\vzeta)^\dag}=\vxi_\vGamma + \Re\vzeta.
$$
Clearly, $D(\vzeta)^\dag \sigma_j^n D(\vzeta)=\sigma_j^n$, hence
$$
\vsigma_{D(\vzeta) \vGamma D(\vzeta)^\dag}=\vsigma_\vGamma.
$$
The internal part of the Hamiltonian in terms of ladder operators is
$$
\vH_0=2\sum_{m=1}^M \left(a_m^\dag a_m + \frac{1}{2}\right) + \sum_{m=1}^M\sum_{n=1}^N \frac{a_m + a_m^\dag}{\sqrt{2}}\Lambda_{mn}\sigma_z^n - \sum_{n=1}^N t_n\sigma_x^n.
$$
Using \cref{displadder}, the individual parts transform as
\begin{align*}
D(\vzeta)^\dag a_m^\dag a_m D(\vzeta) &= (a_m^\dag + \ol{\zeta}_m/ \sqrt{2})(a_m + \zeta_m/ \sqrt{2}) \\
&=a_m^\dag a_m + \zeta_m a_m^\dag/ \sqrt{2}  + \ol{\zeta}_m a_m/ \sqrt{2} + \tfrac{1}{2}|\zeta_m|^2\\
&=a_m^\dag a_m + (\Re \zeta_m)x_m  + (\Im \zeta_m) \i\partial_{x_m} + \tfrac{1}{2}|\zeta_m|^2,\\
D(\vzeta)^\dag \frac{a_m + a_m^\dag}{\sqrt{2}}\Lambda_{mn} \sigma_z^n D(\vzeta) &= \frac{a_m + a_m^\dag}{\sqrt{2}}\Lambda_{mn} \sigma_z^n + (\Re \zeta_m) \Lambda_{mn}\sigma_z^n,\\
D(\vzeta)^\dag t_n\sigma_x^n D(\vzeta) &= t_n\sigma_x^n.
\end{align*}
Consequently, the expectation value of the internal part of the Hamiltonian transforms as
\begin{equation}\label{displham}
\Tr(\vH_0 D(\vzeta) \vGamma D(\vzeta)^\dag)=\Tr \vH_0\vGamma + 2(\Re\vzeta)\cdot\vxi_\vGamma + 2(\Im\vzeta)\cdot \langle \i\grad\rangle_{\vGamma} + (\Re\vzeta)\cdot\vLambda\vsigma_\vGamma + |\vzeta|^2.
\end{equation}
After these preparations, we are ready to prove the stated formula. First, let $\vGamma$ be an optimizer for $F_\LI(\vsigma,\vxi)$, and 
put $D(\vzeta) \vGamma D(\vzeta)^\dag$, $\vzeta\in\RR^M$, as a trial state for $F_\LI(\vsigma,\vxi+\vzeta)$, to get
$$
F_\LI(\vsigma,\vxi+\vzeta)\le F_\LI(\vsigma,\vxi) + 2\vzeta\cdot\vxi + \vzeta \cdot \vLambda\vsigma  +  |\vzeta|^2.
$$
Conversely, let $\vGamma$ be an optimizer for $F_\LI(\vsigma,\vxi+\vzeta)$ and put $D(-\vzeta) \vGamma D(-\vzeta)^\dag$ as a trial state to $F_\LI(\vsigma,\vxi)$
to obtain the opposite inequality.

Next, for \eqref{flprop:item:real}, note that the complex conjugation $\vGamma\mapsto\ol{\vGamma}$ leaves the energy $\Tr\vH_0\vGamma$ invariant,
which can be easily verified by writing it with the integral kernel $\ol{\vGamma(\vx,\vx')}=\vGamma(\vx',\vx)$. Let $\vGamma$ be an optimizer for $F_\LI(\vsigma,\vxi)$ and note that
$(\vGamma+\ol{\vGamma})/2$ is a density matrix as well that has same internal energy $F_\LI(\vsigma,\vxi)$. Since it is real-valued, this proves \eqref{flprop:item:real}.

To show \eqref{flprop:item:virial}, we consider the one-parameter family of scaled density matrices $\vGamma_\vM(\vx,\vx')=(\det\vM)\vGamma(\vM\vx,\vM\vx')$ for any diagonal
matrix with positive entries $\vM=\diag(\mu_1,\ldots,\mu_M)\in\RR^{M\times M}$. If $\vGamma$ is an optimizer for $F_\LI(\vsigma,\bm{0})$ then $\vsigma_{\vGamma_{\vM}}=\vsigma_\vGamma = \vsigma$ and 
$\vxi_{\vGamma_{\vM}}=\vM^{-1}\vxi_{\vGamma}=0$. Moreover, 
$$
\Tr\vH_0\vGamma_{\vM}=\Tr(- \vM^2 (\grad\otimes\grad) \vGamma) + \Tr |\vM^{-1}\vx|^2\vGamma + \Tr (\vM^{-1}\vx\cdot\vLambda\vsigma_z\vGamma) - \Tr (\vt\cdot\vsigma_x\vGamma).
$$
Since $\vGamma$ is an optimizer, $\frac{\partial}{\partial \mu_m} \Tr\vH_0\vGamma_{\vM}|_{\mu_m=1}=0$, i.e.,
$$
0=2\Tr((-\partial_{x_m}^2 - x_m^2)\vGamma) - \Tr( x_m\ve_m\cdot \vLambda\vsigma_z\vGamma)
$$
from which the stated formula follows after summation over $m$.

For \eqref{flprop:item:no-momentum}, let $\vGamma$ be an optimizer for $F_\LI(\vsigma,\vxi)$ and consider for a fixed $m$ the mapping $s\mapsto \vGamma_s=D(\i s\ve_m) \vGamma D(\i s\ve_m)^\dag$. Using \cref{displham},
we find
$$
0=\frac{\dd}{\dd s}\Big|_{s=0} \Tr \vH_0\vGamma_s = 2 \langle \i\partial_{x_m} \rangle_{\vGamma},
$$
as required.
\end{proof}

\begin{proof}[Proof of \cref{prop:force-balance}]
Since $\Tr \vH(\vv,\vj)\vGamma=E(\vv,\vj)$ the $\vGamma$ is the optimizer in \eqref{eq:FLprop:Gamma-inf} and thus from \eqref{eq:FLprop:E-inf} we have by characterizing the optimizer with the subdifferential that
$$
(-\vv,-\vj) \in \underline{\partial} F_\LI (\vsigma_\vGamma,\vxi_\vGamma).
$$
Now put in the displacement rule \eqref{flprop:item:displacement} for $F_\LI (\vsigma_\vGamma,\vxi_\vGamma)$ and from differentiation with respect to $\vxi$ directly get
$$
-\vj = \vLambda\vsigma_\vGamma + 2\vxi_\vGamma,
$$
which concludes the proof. The same result can be achieved as a ``hypervirial theorem'' \cite{hirschfelder1960classical} with respect to the momentum operator $-\i\vnabla$.
\end{proof}

\subsection{Proofs of \cref{specsec}}

\begin{proof}[Proof of \cref{fllprop_N=M=1}]
For \eqref{item:sign-flip}, we use the invertible transformation
$$
\psi(x)\mapsto \wt{\psi}(x)=\begin{pmatrix}
\psi^-(-x)\\
\psi^+(-x)
\end{pmatrix}.
$$
Then $\|\wt{\psi}^+\|^2=\frac{1-\sigma}{2}$, $\|\wt{\psi}^-\|^2=\frac{1+\sigma}{2}$ 
 and $\xi_{\wt{\psi}}=-\xi$ for any $\psi$ such that $\|\psi^+\|^2=\frac{1+\sigma}{2}$,
$\|\psi^-\|^2=\frac{1-\sigma}{2}$ and $\xi_\psi=\xi$. Also, $\dua{\wt{\psi}}{H_0\wt{\psi}}=\dua{\psi}{H_0\psi}$, from which the claim follows.

As for the real-valuedness part of \eqref{item:LL-prop-non-neg}, it is enough to see that the real and imaginary parts of $\psi$ decouple in the expression of the quadratic form $\dua{\psi}{H\psi}$, and that the minimization can be carried out for the real and imaginary parts separately. This is obvious, except for the term $-t\langle\sigma_x\rangle_\psi = -2t\Re\dua{\psi^+}{\psi^-}$. But if we write for any admissible $\psi$,
$$
2\Re\dua{\psi^+}{\psi^-}= 1 - \|\psi^+ - \psi^-\|^2=1 - \|\Re\psi^+ -\Re\psi^-\|^2 - \|\Im\psi^+ -\Im\psi^-\|^2
$$
using the polarization identity, then real-valuedness follows.
To see non-negativity, let $\psi$ be an admissible wavefunction. Define the level sets $A_+^\alpha=\{x\in\RR :  \psi^\alpha(x)\ge 0\}$ and $A_-^\alpha=\{x\in\RR :  \psi^\alpha(x)<0\}$. Set
$$
\wt{\psi}=\begin{cases}
\psi & A_+^+\cap A_+^- \\
(\psi^+,-\psi^-) & A_+^+\cap A_-^-\\
(-\psi^+,\psi^-) & A_-^+\cap A_+^-\\
-\psi & A_-^+\cap A_-^-
\end{cases}
$$
which is non-negative.
It is clear that the constraints and all the terms in $\dua{\wt{\psi}}{H_0\wt{\psi}}$ are unchanged except for $-2t\Re\dua{\wt{\psi}^+}{\wt{\psi}^-}$, which, again by the polarization identity, can be reduced to looking at
\begin{align*}
\|\wt{\psi}^+ - \wt{\psi}^-\|^2 &=
\int_{A_+^+\cap A_+^-} |\psi^+ - \psi^-|^2 + \int_{A_+^+\cap A_-^-} |\psi^+ + \psi^-|^2 \\
&+ \int_{A_-^+\cap A_+^-} |\psi^+ + \psi^-|^2
+ \int_{A_-^+\cap A_-^-} |\psi^+ - \psi^-|^2.
\end{align*}
Here, the middle two terms did not increase (now these integrands possibly contain a sum of a positive and a negative function) and the other two integrals remained invariant. We deduce that the transformation $\psi\mapsto\wt{\psi}$ did not increase the energy.

To see \eqref{item:N=1,M=1,VR}, consider the one-parameter family of wavefunctions $\psi_\mu(x)=\sqrt{\mu}\psi(\mu x)$, which has $\|\psi_\mu^+\|^2=\frac{1+\sigma}{2}$,
$\|\psi_\mu^-\|^2=\frac{1-\sigma}{2}$ and $\xi_{\psi_\mu}=\frac{1}{\mu}\xi_\psi$. Moreover,
$$
\dua{\psi_\mu}{H_0\psi_\mu}=\int \left[ \mu^2 |\psi'|^2 + \frac{x^2}{\mu^2} |\psi|^2 \right] \dd x +
\frac{\lambda}{\mu} \int x\left( |\psi^+|^2 - |\psi^-|^2\right)  \dd x -2t\Re\dua{\psi^+}{\psi^-}.
$$
But $\frac{\dd}{\dd\mu}|_{\mu=1}\dua{\psi_\mu}{H_0\psi_\mu}=0$ which is equivalent to the stated relation if we further use $\int x\left(|\psi^+|^2 + |\psi^-|^2\right) \dd x = \xi_\psi = \xi$.

Next, for \eqref{item:N=1,M=1,1stmoment} we use the transformation
\begin{equation*}
    \psi\mapsto \psi_s = \begin{pmatrix}
    \psi^+(x-(1-\sigma)s)\\
    \psi^-(x+(1+\sigma)s)
    \end{pmatrix}
\end{equation*}
which is chosen such that it keeps
$\sigma_{\psi_s} = \sigma_{\psi} = \sigma$, $\xi_{\psi_s} = \xi_{\psi} = \xi$
(not necessarily zero) constant. We consider the derivative $\partial_s \langle H_0 \rangle_{\psi_s}$ that must be zero at $s=0$ if $\psi$ is an optimizer. The parts of $H_0$ are,
\begin{align*}
    \langle -t\sigma_x \rangle_{\psi_s} &= -2t \int \psi^+(x-(1-\sigma)s) \psi^-(x+(1-\sigma)s) \,\dd x,\\
    \langle \lambda \sigma_z x \rangle_{\psi_s} &= \lambda \langle x+(1-\sigma)s \rangle_{\psi^+} - \lambda \langle x-(1+\sigma)s \rangle_{\psi^-} \nonumber\\
      &= \lambda (\langle x \rangle_{\psi^+} - \langle x \rangle_{\psi^-}) + \lambda s (1-\sigma^2),\\
    \langle p^2 \rangle_{\psi_s} &= \langle p^2 \rangle_{\psi}, \\
    \langle x^2 \rangle_{\psi_s} &= \langle (x+(1-\sigma)s)^2 \rangle_{\psi^+} + \langle (x-(1+\sigma)s)^2 \rangle_{\psi^-} \nonumber\\
      &= \langle x^2 \rangle_{\psi} + 2s(\langle x \rangle_{\psi^+} - \langle x \rangle_{\psi^-}) - 2s\sigma\xi + s^2(1+\sigma^2).
\end{align*}
Now the derivative at $s=0$ yields
\begin{equation*}\begin{aligned}
    \frac{\dd}{\dd s} \Big|_{s=0} \langle H_0 \rangle_{\psi_s} = 2&t(1-\sigma) \int (\psi^+)' \psi^- - 2t(1+\sigma) \int \psi^+ (\psi^-)'  \\
    &+ \lambda(1-\sigma^2) + 2 \langle x \rangle_{\psi^+} - 2 \langle x \rangle_{\psi^-} - 2 \sigma\xi = 0.
\end{aligned}\end{equation*}
As $\langle x \rangle_{\psi^-}  = \xi- \langle x \rangle_{\psi^+} $, this yields the stated result.
Moreover, if we similarly consider $\frac{\dd^2}{\dd s^2}\big|_{s=0} \langle\psi_s, H_0 \psi_s\rangle \geq 0$ we find the stated inequality.

Finally, \eqref{item:decoupling-FLL-explicit} is just a special case of \cref{lemma:FLL0}.
\end{proof}

\begin{proof}[Proof of \cref{thrm:EulerLagrange}]
Existence of the Lagrange multipliers in the regular case $\sigma\neq\pm 1$ was already treated in \cref{elthm}, and uniqueness follows from the fact that the functions $\psi$, $\sigma_z\psi$ and $x\psi$ are linearly independent.

It remains to consider the irregular case $\sigma=\pm 1$. We only look at $\sigma=+1$, the other case is analogous, and proceed similarly as in the proof of \cref{elthm}. The criticality condition reads $\langle \chi, H_0\psi \rangle = 0$ with $\psi = (\psi^+,0)^\top$, $\chi = (\chi^+,0)^\top$, $\int\chi^+\psi^+=0$ and $\int x\chi^+\psi^+=0$.
This implies, by a similar argument as before, that
$$
-(\psi^+)'' + x^2 \psi^+ + \lambda x \psi^+= \mu^+\psi^+ + \mu^0 x\psi^+
$$
for some unique $\mu^+,\mu^0\in\RR$. Rearranging, we get
$$
-(\psi^+)'' + \left(x+\frac{\lambda -\mu^0}{2}\right)^2 \psi^+ = \left(\mu^+ + \frac{(\lambda-\mu^0)^2}{4}\right) \psi^+.
$$
The constraint $\xi_{\psi}=\xi$ implies that $\xi=-(\lambda - \mu^0)/2$, or $\mu^0=\lambda+2\xi$.
From the eigenvalues of the harmonic oscillator, we find $\mu^+=2n+1-\xi^2$. The stated equation follows by translation.
\end{proof}

\begin{proof}[Proof of \cref{uniquevrep}]
We know from \cref{fllprop_N=M=1}~\eqref{item:LL-prop-non-neg} that there is at least one optimizer $\psi$ of $F_\LL(\sigma,\xi)$ which is non-negative.
Since for $\sigma\neq\pm 1$ this $\psi$ satisfies the Schr\"odinger equation \cref{eleqs} with $H(v,j)$ by \cref{thrm:EulerLagrange}, it is even positive a.e.\ by the strong UCP (\cref{ucpthm}). We also have $\exp(-tH(v,j))$ as a positivity improving operator~\cite{Hirokawa2014,Nguyen2024}, which together with steps 3 and 4 in \cite{uniqueness} means that the ground-state eigenvector of $H(v,j)$ is strictly positive and non-degenerate. Any excited eigenstate is orthogonal to the ground state and consequently must change sign. Thus the optimizer $\psi$ must also be the unique ground state.

For the uniqueness part, suppose that $\wt{\psi}$ is another optimizer of $F_\LL(\sigma,\xi)$ with potentials $(\wt{v},\wt{j})$, then 
$E(v,j)=E(\wt{v},\wt{j})+\sigma(v-\wt{v}) + \xi(j-\wt{j})$. We have
\begin{align*}
E(v,j)&\le \dua{\wt{\psi}}{H(v,j)\wt{\psi}}=E(\wt{v},\wt{j}) + \sigma(v-\wt{v}) + \xi(j-\wt{j})=E(v,j),
\end{align*}
hence there is equality. Therefore, $\dua{\wt{\psi}}{H(v,j)\wt{\psi}}=E(v,j)$, so $\wt{\psi}$ satisfies $H(v,j)\wt{\psi}=E(v,j)\wt{\psi}$,
which contradicts the fact that the ground state of $H(v,j)$ is unique. 

Finally, if $\sigma=\pm 1$ then note that by \cref{thrm:EulerLagrange} the $\psi$ corresponding to $(\sigma,\xi)$ must satisfy the Schr\"odinger equation for the harmonic oscillator \cref{eleqsharm} and is thus unique. But \cref{eleqsharm} is not a Schr\"odinger equation of the form $H(v,j)\psi=E\psi$. More generally, if $H(v,j)\psi=E\psi$ for some $(v,j)$ and $\psi$,
then $\psi^{\mp}=0$ implies that $\psi^{\pm}=0$ from the coupling term $-t\sigma_x$, $t\neq 0$, so $\psi\equiv 0$,
hence $v$-representability does not hold in this case.
\end{proof}

\begin{proof}[Proof of \cref{prop:F-diff-N=M=1}]
We start by showing $F_\LI=F_\LL$ on $(-1,1)\times\RR$. By \cref{uniquevrep}~\eqref{uniquevrep:item:inner} the (unique) optimizer of $F_\LL(\sigma,\xi)$ is a ground state of $H(v,j)\psi=E(v,j)\psi$ with $\sigma_\psi=\sigma$ and $\xi_\psi=\xi$. Any ensemble state $\Gamma$ can be written as a convex combination over pure-state projectors onto normalized $\psi_j$, i.e., $\Gamma=\sum_j c_j |\psi_j\rangle\langle\psi_j|$ with $\sum_j c_j =1$. With the variational principle for the ground state $\psi$ this means
\begin{align*}
E(v,j) &= F_\LL(\sigma,\xi) + \sigma v + \xi j = \sum_j c_j \langle \psi,H(v,j)\psi \rangle \\
&\leq \sum_j c_j \langle \psi_j,H(v,j)\psi_j \rangle = \Tr H(v,j)\Gamma = \Tr H_0\Gamma + \sigma v + \xi j.
\end{align*}
Now taking the infimum over all ensemble states $\Gamma$ with $\sigma_\Gamma=\sigma$ and $\xi_\Gamma=\xi$ yields $F_\LL\leq F_\LI$. Since anyway $F_\LI\leq F_\LL$ we have that $F_\LI=F_\LL$ on $(-1,1)\times\RR$. Since further the $(v,j)$ such that $\psi$ is a ground state of $H(v,j)\psi=E(v,j)\psi$ are unique, the subdifferential $\partial F_\LI(\sigma,\xi)$ is a singleton and thus the functional is automatically differentiable on $(-1,1)\times\RR$.

The cases $\sigma=\pm 1$ are now shown separately. First, by \cref{flprop}~\eqref{flprop:item:displacement} shift to $F_\LI(\sigma,\xi) = F_\LI(\sigma,0) + \lambda\sigma\xi + \xi^2$, then use the property that $F_\LI(\sigma,0)$ can be evaluated as the infimum over all convex combinations of $F_\LI(\sigma_j,0)$ with $\sum_j c_j\sigma_j=\sigma$ and $\sum_j c_j = 1$ \cite[Eq.~(4.6)]{Lieb1983}. But if $\sigma=\pm 1$ then also all $\sigma_j=\pm 1$, so no non-trivial convex combination is possible and we directly have $F_\LI(\sigma,0) = F_\LL(\sigma,0)$ from which follows $F_\LI(\sigma,\xi) = F_\LL(\sigma,\xi)$ by the displacement rule.
\end{proof}



\end{document}